\documentclass[conference,letterpaper]{IEEEtran}
\usepackage{mathtools,amssymb,amsthm}
\usepackage{graphicx}
\usepackage{flushend}
\usepackage{tikz}
\usetikzlibrary{arrows.meta,positioning}
\usepackage{url}
\usepackage[hidelinks]{hyperref}
\usepackage{xurl}
\newtheorem{theorem}{Theorem}[section]
\newtheorem{lemma}[theorem]{Lemma}
\newtheorem{proposition}[theorem]{Proposition}
\newtheorem{corollary}[theorem]{Corollary}
\theoremstyle{definition}
\newtheorem{definition}[theorem]{Definition}
\newcommand{\Z}{\mathbb Z}

\newcommand{\F}{\mathbb F}
\newcommand{\rt}{\mathrm{route}}
\newcommand{\sym}{\mathrm{sym}}
\newcommand{\lin}{\mathrm{lin}}
\newcommand{\dist}{\operatorname{dist}}

\newcommand{\Id}{\mathrm I}
\newcommand{\cl}{\operatorname{cl}}
\title{Network Coding Can Beat Routing\\ in Undirected Multiple-Unicast Networks\\[2pt] {\Large A Formally Verified Counterexample to the Li--Li Conjecture}}
\author{\IEEEauthorblockN{Xindan Zhang\IEEEauthorrefmark{1}, Baochun Li\IEEEauthorrefmark{1}, Zongpeng Li\IEEEauthorrefmark{2}}
\IEEEauthorblockA{\IEEEauthorrefmark{1}University of Toronto\\
\IEEEauthorrefmark{2}Tsinghua University}}
\hypersetup{pdftitle={Network Coding Can Beat Routing in Undirected Multiple-Unicast Networks: A Formally Verified Counterexample to the Li-Li Conjecture},pdfauthor={Xindan Zhang, Baochun Li, Zongpeng Li},pdfsubject={Undirected multiple-unicast network coding},pdfkeywords={network coding, fractional routing, finite abelian groups, high girth, bounded degree, formal verification}}
\begin{document}
\bstctlcite{IEEEexample:BSTcontrol}
\maketitle
\begin{abstract}
Network coding lets the nodes of a network combine messages rather than merely forward them. In undirected networks, where the two directions of an edge share its capacity, Li and Li conjectured in 2004 that coding offers no advantage over fractional routing; the conjecture has been confirmed for many special classes but never settled in general. In this paper, we disprove it by constructing a finite connected simple undirected network with unit shared edge capacities, maximum degree three, and distinct leaf terminals, on which a binary linear block code achieves a common rate strictly above the maximum fractional routing rate. The construction turns a short completion-time code into a reversible circuit whose registers all carry independent messages, and a temporal metric on its wires yields the strict routing bound. The underlying integer-coefficient construction works over every finite field and every nontrivial finite abelian group; for deterministic fixed-schedule codes the common-rate supremum is one and the closure of the rate region is the unit cube. Amplifying the gap with a degree-preserving tensor construction and an even subdivision gives, for every $0<\beta<1$, an unbounded family of connected, simple, subcubic, bipartite networks with girth at least $n^\beta$ on which coding approaches rate one while routing is bounded by $K_\beta/(\log n)^c$, with one positive exponent $c$ independent of $\beta$. The finite counterexample and the family theorem are formalized in Lean.
\end{abstract}

\section{Introduction}\label{sec:introduction}
\begin{figure*}[t]
\centering
\input{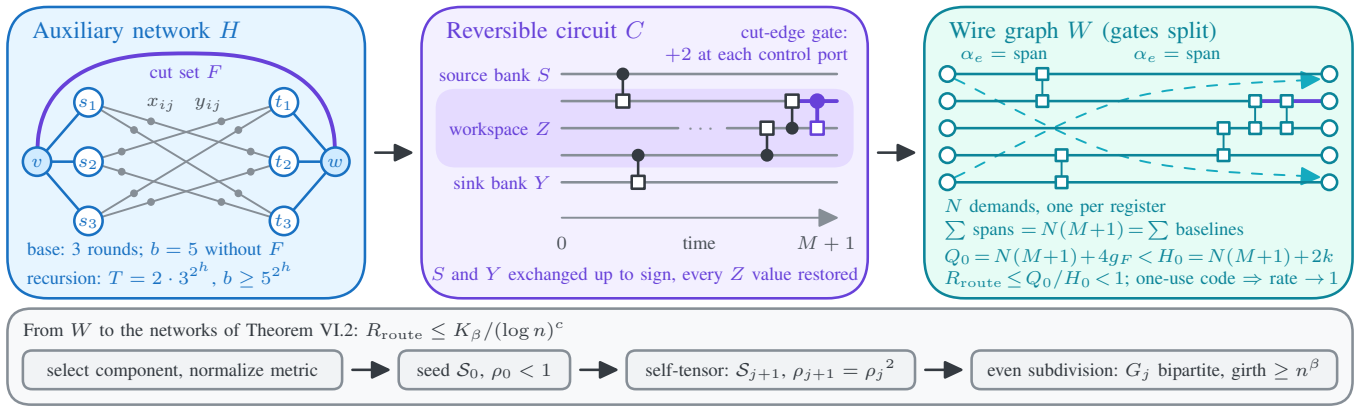}
\caption{An overview of the construction. The left panel shows the base gadget of the auxiliary network, with three terminal pairs drawn; recursive composition produces the instance $H$ that is compiled into the reversible circuit $C$, whose wires become the physical edges of the wire graph $W$, with one unicast demand per register. The circuit and the wire graph are schematic, with most gates omitted. In $C$, a dot marks the control and a square the target of a gate; in $W$, drawn after the gate splitting of Section~\ref{sec:split}, each gate is two vertices joined by its internal edge, and the dashed curved arrows indicate demand pairs, not physical edges. Giving every register a demand balances the total temporal wire length $N(M+1)$ against the baseline demand distances. The metric penalty of two at the control ports of the $g_F$ cut-edge gates adds $4g_F$ to the edge budget $Q_0$, while the $k$ designated demands add $2k$ to the demand bound $H_0$; hence $Q_0<H_0$, every fractional routing rate is at most $Q_0/H_0<1$, and the one-use code pipelines toward rate one. Component selection and metric normalization produce the seed; self-tensoring squares its routing ceiling, and even subdivision preserves the ceiling while making the networks bipartite with polynomial girth.}
\label{fig:overview}
\end{figure*}
In an undirected network, an edge can carry a message in either direction, but its two directions share one capacity budget. With several independent source--sink pairs, routing sends each message along paths, and fractional routing may split each session's traffic among all available paths with arbitrary nonnegative real weights. Network coding, in addition, allows intermediate vertices to combine the messages they receive~\cite{ACLY2000}. Whether this additional freedom can raise the achievable rates is the question of this paper.

In 2004, Li and Li conjectured that it cannot: every rate vector achievable by network coding is also achievable by fractional routing~\cite{LiLi2004,LiLiCISS2004}. Harvey, Kleinberg, and Rasala Lehman formulated the same conjecture independently~\cite{HKL2004}. Over the two decades since, it has been confirmed for one class of networks after another. It holds wherever cuts already determine the routing rate, as for two sessions and for planar networks with all terminals on one face~\cite{Hu1963,OS1981}. It holds for the Okamura--Seymour example~\cite{OS1981}, settled concurrently by Harvey, Kleinberg, and Rasala Lehman, by Jain, Vazirani, and Yuval, and by Kramer and Savari~\cite{HKL2006,JVY2006,KS2006}, and for bipartite families generalizing it~\cite{HKL2006,JVY2006}. It also holds for further bipartite classes~\cite{AY2008}, for classes on which a partition bound is tight~\cite{TQ2019,QT2022}, for networks with at most six coding nodes~\cite{Yin2018}, and for networks with few terminal locations, a restricted planar structure, or further topological restrictions~\cite{Liu2026,LiuINFOCOM2026}. Xiahou, Li, Wu, and Huang recast the problem geometrically~\cite{Xiahou2014}, Langberg and M\'edard relate the problem to multicast: when coding lets every source multicast to all sinks at rate $r$, routing serves the unicast demands at rate $r/3$~\cite{LM2009}; and Liu, Li, Fan, and Chen reduce the conjecture to the independence of irreducible cores~\cite{Liu2026sif}. The conjecture has also served as a hypothesis: lower bounds for external-memory integer sorting and for multiplication circuits were proved under it~\cite{FHLS2019,AFKL2019}.

But is the conjecture true? What the general theory provides is an upper bound. The coding rate never exceeds the sparsity of the network, the minimum ratio of cut capacity to separated demand over cuts with positive separated demand~\cite{HKL2006}, and for $k\geq2$ sessions the maximum concurrent flow is within a factor $O(\log k)$ of sparsity~\cite{LLR1995,AR1998}. Coding can therefore beat routing by at most a logarithmic factor. For multi-source multicast in undirected networks, a coding advantage of $\Omega(\log k)$ was shown only recently, and its authors note that for unicast sessions it was not even known whether coding helps at all~\cite{BH2025}. Braverman, Garg, and Schvartzman (BGS) showed that if any undirected network admits a strict coding advantage, a tensor-product construction amplifies it into a polylogarithmic one~\cite{BGS2017}. Whether coding ever helps in undirected networks, and whether it then helps polylogarithmically, therefore comes down to one finite instance.

In this paper, we construct that instance, and the Li--Li conjecture is false. We present a finite connected simple undirected network with unit shared edge capacities, maximum degree three, and pairwise distinct leaf terminals, on which a binary linear block code delivers every session at a common rate strictly above the best possible fractional routing rate. The separation is established at a finite block length, before any capacity limit is taken, and it is certified by a distance argument that bounds every fractional routing over all simple paths with all nonnegative real weights. The network therefore lies outside each of the classes listed above.

Our starting point is a separation for a different objective. Haeupler, Wajc, and Zuzic (HWZ) constructed networks with a gap between coding and routing completion times: their examples combine a short code with long source--sink distances once a small set of edges, the cut set, is removed~\cite{HWZ2020}. A gap in completion time, however, does not by itself give a throughput separation, because a long path can still carry a pipeline at high rate. Our key idea is to turn the short code into a \emph{reversible computation}, and to build a new network whose physical edges are the wires of that computation (Fig.~\ref{fig:overview}). Every register of the computation, including every workspace register, becomes an independent unicast demand; the computation exchanges the two logical banks, up to a sign removed locally, and restores every workspace register to whatever value it started with.

Why does making every register a demand help? Assign each wire the span of circuit time that it bridges. Summed over all wires, these spans equal exactly the sum of the baseline distances of all demands, where the baseline is one full circuit duration, the least that any input-to-output route can cost. The $k$ designated demands, those carrying the source messages to the sinks, cannot be served at the baseline: a route that only moves forward in time projects to a walk in the auxiliary graph that is too short to avoid the cut set, so it must enter or leave, through its control port, a gate that simulates a transmission on a cut edge. We charge two units at every such port; a backward step already costs at least two more by itself. The charges add at most $48Tf$ to the wire budget, where $T$ is the duration of the auxiliary code and $f$ the size of its cut set, while the designated demands add $2k$ to the demand distance; the auxiliary code is chosen with $k>24Tf$, so every common routing rate is strictly below one. After the gates are split and the metric is normalized, the code uses each physical edge exactly once per execution, so independent executions can start one round apart, and a sufficiently long but finite pipeline attains a rate above the routing bound.

The construction is defined over the integers, and the same integer instructions can be interpreted in every finite field and every nontrivial finite abelian group. The same graph, schedule, and coefficients therefore separate coding from routing over all of these alphabets at once; over fields, the codes are linear. To take the separation a step further, we ask how large the coding advantage can become, and under which structural restrictions it survives. We do not use the amplification theorem of BGS as a black box: we specialize their tensor construction so that the integer code, the degree bound, and the leaf terminals survive each iteration, and we prove a quantitative bound on the size of the incidence graphs that drive it. For each fixed $0<\beta<1$, a final even subdivision, whose length depends on $\beta$, makes every network bipartite with girth at least $n^\beta$, where $n$ is the number of vertices. A longer subdivision even yields a single family with girth $n^{1-o(1)}$, at the cost of halving the exponent (Corollary~\ref{cor:nearlinear}).

Highlights of our original contributions in this paper are three-fold. \emph{First}, we give the first counterexample to the Li--Li conjecture: a complete finite construction of the graph, a causal binary linear block code, and a strict upper bound on fractional routing, with the separation established at a finite block length (Theorem~\ref{thm:finite}).
\emph{Second}, we show that the separation does not depend on the alphabet, and we determine the coding capacity of the counterexample and of every network in the families below: for deterministic fixed-schedule codes, the common-rate supremum is one, and the closure of the vector-rate region is the unit cube (Proposition~\ref{prop:capacity}). \emph{Finally}, for every $0<\beta<1$ we construct an unbounded family of connected, simple, subcubic, bipartite networks with distinct leaf terminals and girth at least $n^\beta$, on which routing is bounded by $K_\beta/(\log n)^c$ with a single positive exponent $c$ independent of $\beta$. The graphs and integer codes are fixed before any alphabet is chosen, so coding approaches rate one over every finite field and every nontrivial finite abelian group (Theorem~\ref{thm:main}).

Both theorems are formalized in Lean~\cite{Lean4} using Mathlib~\cite{Mathlib2020}.\footnote{The Lean proofs are available at \url{https://github.com/iqua/li-li-conjecture}. OpenAI's GPT-6 Astra assisted in developing the proofs and the Lean code.} The development relies on Lean's standard axioms only, and the final appendix explains how its organization mirrors the proof.

\section{Background and Problem Formulation}\label{sec:model}
This section fixes the two communication models that the theorems compare and the length argument that bounds routing. A network consists of a finite undirected graph $G=(V,E)$ and $k>0$ ordered pairs $(s_i,t_i)$ with $s_i\ne t_i$. Some intermediate graphs have parallel edges; the final graphs are simple. Each physical edge has unit capacity shared between its two directions. Vertices may use different incident edges simultaneously; local storage and computation are free. A graph is subcubic if every vertex has degree at most three. Its girth is the length of its shortest cycle, taken as infinite when there is no cycle. All unqualified logarithms are natural.

\subsection{Coding protocols and normalized rates}
Fix a finite alphabet $A$ with at least two symbols. Session $i$ initially holds an arbitrary message $X_i\in A^{B_i}$ at $s_i$. The messages vary independently: correctness is required for every tuple of their values. A fixed finite schedule specifies events, each a directed use of a physical edge at a round in $\{0,\ldots,D-1\}$. No two events use the same physical edge in the same round. A transmitted symbol is a deterministic function of the sender's own message coordinates and symbols received in strictly earlier rounds. A sink must recover its message exactly from its local information. The schedule, including its directions and idle slots, is independent of the messages.

The achieved rate vector is $(B_1/D,\ldots,B_k/D)$ for $D>0$. The set of such vectors is $\mathcal R_\sym(G;A)$. Zero-length messages are allowed. If $A=\F$ is a finite field, $\mathcal R_\lin(G;\F)$ restricts every local encoding and decoding function to an $\F$-linear map. These maps may use arbitrary coefficients in $\F$. A common rate has $B_i=B$ for every session; its supremum is the common-rate capacity in the specified class.

The normalization is by alphabet symbols. Choose each input symbol uniformly and independently and measure entropy in units of $\log_2|A|$ bits. Session $i$ then has source entropy $B_i$. Order the events by round, breaking ties arbitrarily. Every predecessor used by an event is in a strictly earlier round, so this order gives a legal causal transcript. If an edge has $d_+$ scheduled uses in one direction and $d_-$ in the other, its two directional transcript entropies are at most $d_+$ and $d_-$, and $d_++d_-\leq D$. Dividing all entropies by $D$ therefore gives coding-feasible rates under the usual shared-capacity undirected model~\cite{LiLi2004,HKL2006,BGS2017}. Repeating a finite protocol on fresh independent blocks also gives its stated long-run rate.

\begin{definition}\label{def:integer}
An integer-coefficient protocol specifies each transmitted symbol and each decoded coordinate as an integer linear combination of locally available input coordinates and strictly earlier received symbols. Its correctness is an identity of integer coefficient rows.
\end{definition}

If $A$ is an additive abelian group, interpret an integer coefficient by repeated addition and negation. An identity over $\Z$ then holds in $A$. Over a field, interpret integers through the canonical map $\Z\to\F$; every local function is $\F$-linear. This interpretation concerns the local formulas, not only the final transfer matrix. It therefore preserves causality and locality as well as exact decoding.

\subsection{Fractional routing and distance certificates}
For each pair let $\mathcal P_i$ be the set of all simple undirected paths from $s_i$ to $t_i$. A common rate $r\geq0$ is feasible if there are real numbers $x_{i,p}\geq0$ such that
\begin{equation}\label{eq:routing}
 \sum_{p\in\mathcal P_i}x_{i,p}\geq r,
 \qquad
 \sum_i\sum_{\substack{p\in\mathcal P_i\\e\in p}}x_{i,p}\leq1
 \quad(e\in E).
\end{equation}
The second constraint counts both directions together. A simple path visits no vertex twice and therefore has at most $|V|-1$ edges. Since the vertex and edge sets are finite, every $\mathcal P_i$ is finite. Each path has at least one edge because $s_i\ne t_i$, so its weight is at most one under the load constraints. The feasible path-weight vectors form a closed bounded subset of a finite-dimensional real space. The largest common rate therefore exists; denote it by $R_\rt(G)$.

Allowing walks gives the same maximum. Whenever a walk repeats a vertex, delete the portion between two consecutive occurrences of that vertex. This shortens the walk and increases no physical-edge multiplicity. Repeating the operation terminates at a simple path with the same endpoints. For arbitrary nonnegative real walk weights with well-defined loads, aggregate the weights of walks that reduce to the same simple path. Nonnegativity allows this regrouping, each session's total weight is preserved, and no edge load increases. Thus the complete simple-path formulation captures ordinary fractional routing.

Suppose $\ell_e\geq0$ and every path of session $i$ has $\ell$-length at least $h_i\geq0$. Finite summation gives
\begin{equation}\label{eq:lengthbound}
 r\sum_i h_i
 \leq\sum_{i,p}x_{i,p}\sum_{e\in p}\ell_e
 =\sum_e\ell_e\operatorname{load}(e)
 \leq\sum_e\ell_e.
\end{equation}
In particular, call a walk from a source $s_i$ to its own sink $t_i$ an own-pair walk; if every own-pair walk has at least $b>0$ edges, then
\begin{equation}\label{eq:hopceiling}
 R_\rt(G)\leq\frac{|E|}{kb}.
\end{equation}
All subsequent routing arguments bound the lengths of arbitrary undirected walks from below. No direction or timing of a coding schedule restricts the path universe in~\eqref{eq:routing}.

\section{A Finite Counterexample}\label{sec:core}
We first prove the finite separation on which the rest of the paper builds.

\begin{theorem}[Finite counterexample]\label{thm:finite}
There is a finite connected simple undirected network $G$ with unit shared edge capacities, maximum degree three, and pairwise distinct leaf terminals, and there are positive integers $B,D$, such that an exact causal binary linear code delivers $B$ symbols per session in $D$ rounds and
\begin{equation}\label{eq:finitecounterexample}
 \frac{B}{D}>R_\rt(G).
\end{equation}
\end{theorem}

The construction obtains a short auxiliary code, implements it reversibly, turns every register into a unicast demand, and normalizes a strict routing certificate. We then give an explicit finite block length proving~\eqref{eq:finitecounterexample}. The following representation records the graph and scalar code produced along the way.

\begin{definition}[Scalar instance]\label{def:scalar}
A scalar instance consists of a connected simple undirected graph $G$, positive integers $k,m,b,T$ with $m=|E(G)|$, and $k$ source--sink pairs, such that:
\begin{enumerate}
\item the graph has maximum degree at most three, the $2k$ terminals are distinct leaves, and every edge has at least one nonterminal endpoint;
\item every own-pair undirected walk has at least $b$ edges;
\item there is an exact $T$-round integer scalar code, with one input symbol per session, whose events use each physical edge exactly once.
\end{enumerate}
Each edge is given the reference direction of its unique scalar-code event. Write $\rho(G)=m/(kb)$ for its certified routing ceiling.
\end{definition}

The nonterminal-endpoint condition distinguishes edges from different inner copies after tensoring. It is easy to obtain by subdivision. The one-use condition is stronger than a per-round capacity constraint; it prevents collisions when scalar executions are started one round apart.

Our first target is a scalar instance with $\rho<1$. Its one-use code will support a finite pipeline whose rate exceeds $\rho$. We construct the graph, integer code, and strict distance certificate together, then finish the proof of Theorem~\ref{thm:finite} in Section~\ref{sec:finite-pipeline}.

\subsection{The auxiliary completion-time code}\label{sec:input}
The seed construction needs a finite auxiliary graph $H$ with many demands, a small edge set $F$, and a code that is short compared with the own-pair distances after deleting $F$. We call $F$ the cut set and its elements cut edges, although $F$ need not disconnect the graph; the other edges are noncut.

The recursion uses an incidence graph, the cover defined next, to specify how copies are attached. Its local colors assign one inner session to each outer edge direction, and its girth prevents a short route from bypassing those sessions. The covers are the colored bipartite graphs of BGS and HWZ~\cite{BGS2017,HWZ2020}; Appendix~\ref{app:covers} constructs them and proves a size bound.

\begin{definition}\label{def:cover}
A complementary colored cover with parameters $(d,k)$ is a finite simple bipartite graph with left and right vertex sets $V_L,V_R$, left types in $[k]$, and right types in $[d]$. At each left vertex, exactly one incident edge leads to each right type. At each right vertex, exactly one incident edge leads to each left type. An incidence has colors $(a,i)$ given by its right and left types.
\end{definition}

Call the left vertices of type $i$ the left fiber of $i$, and the right vertices of type $a$ the right fiber of $a$. Each $(a,i)$ matching is a bijection between these two fibers, so when all fibers are nonempty they have one common size $\phi$. Thus
\begin{equation}\label{eq:covercounts}
 |V_L|=k\phi,\qquad |V_R|=d\phi.
\end{equation}
A nonempty connected component of any such cover contains all types because $d,k>0$, so it retains these conditions.

\begin{lemma}[Bounded connected cover]\label{lem:cover}
For $d,k>0$ and an integer $g\geq0$, there is a connected complementary colored cover of girth at least $g$ with
\begin{equation}\label{eq:coverbound}
 |V_L|\leq k\nu,\qquad |V_R|\leq d\nu,\qquad
 \nu=2(\max\{d,k\}+1)^g+1.
\end{equation}
\end{lemma}

\begin{lemma}[Integer completion-time input]\label{lem:input}
For every integer $h\geq0$ and rational $u\geq5$, there is a finite connected loopless undirected multigraph $H$ with $k$ source--sink pairs whose terminals are pairwise distinct, and a nonempty set $F$ of $f$ edges, such that
\begin{equation}\label{eq:input}
 k\geq uf,\qquad
 \dist_{H\setminus F}(s_i,t_i)\geq5^{2^h}.
\end{equation}
The graph $H$ has an exact integer scalar code of duration $3^{2^h}$ under a full-duplex budget of one symbol per direction per round, and of duration $T=2\cdot3^{2^h}$ under a shared unit budget.
\end{lemma}

The construction follows HWZ~\cite{HWZ2020}. Its basic coding operation is already valid over $\Z$. For $q=\lceil u\rceil$, place $q$ sources around a vertex $v$, $q$ sinks around a vertex $w$, and add the edge $vw$. For each $i\ne j$, add a private three-edge path from $s_i$ to $t_j$. In three rounds, the middle edge carries $\sum_i X_i$ and each side path carries its source symbol. Sink $t_i$ recovers
\begin{equation}\label{eq:basecode}
 X_i=\sum_jX_j-\sum_{j\ne i}X_j.
\end{equation}
After deleting $vw$, the distance from $s_i$ to $t_i$ is five. One side path cannot directly connect its own pair, so at least two additional star edges are needed; three side paths already cost nine.

The recursive product replaces the directed uses of nondeleted outer edges by sessions in inner copies, arranged according to a colored high-girth incidence graph. An outer round is simulated by one complete inner execution. The code times and deleted distances multiply, while the target ratio $k/f$ is maintained by choosing enough inner sessions. Identity~\eqref{eq:basecode} and all subsequent substitution identities hold for correlated payloads as well as independent inputs. Appendix~\ref{app:input} gives the complete construction, count formulas, and recursion; Appendix~\ref{app:excursions} proves the distance multiplication. Lemma~\ref{lem:cover} supplies the required finite incidence graph at every stage.

For the seed we take $h=3$ and $u=5^8$. Then
\begin{equation}\label{eq:seednumbers}
 T=13122,\qquad b=390625,\qquad k\geq390625f,
\end{equation}
and consequently
\begin{equation}\label{eq:margins}
 24T=314928<b,\qquad 24Tf<k.
\end{equation}
Only these two strict margins are used below.

\subsection{A reversible integer circuit}\label{sec:compiler}
For every input, allocate a proxy register at its source. For every transmission event, allocate a fresh sender register and a fresh receiver register at the event's two endpoints. Let $Z$ be this workspace register set. Before an event's transmission, add its local integer expression into its sender register using gates
\[
 E_{x\to y}(a):\quad y\leftarrow y+a x,\qquad a\in\Z.
\]
The registers $x,y$ of a gate are distinct; $x$ is the control and $y$ the target. Such a gate is invertible: its inverse has coefficient $-a$. Local gates use registers hosted at the same auxiliary vertex. A cross gate adds a sender register to its receiver register along the event's physical edge.

Events in one auxiliary round have private register pairs. Their cross gates therefore form a phase of disjoint gates, using each physical edge at most once. Execute all local preparations separately before that phase. The complete word defines an invertible integer map $U$ on the workspace.

Let $J$ inject the auxiliary source vector into its proxy registers and put zero in all other workspace coordinates. We check the simulation on $J\xi$ by induction on auxiliary rounds. The source proxies never change. Before an event in round $t$, its private sender and receiver registers have not yet been targets of any gate, so both still contain zero. Its local preparation uses only source proxies and receiver registers for strictly earlier events at that sender. By induction those registers contain precisely the payloads available in the auxiliary code. The preparation therefore puts the intended outgoing payload into the fresh sender register, and the cross gate copies it into the fresh receiver register. Disjointness permits every event of round $t$ to be handled in the same cross phase.

After all rounds, each auxiliary sink has its source proxies and received payloads in registers at that host. Its local decoder gives the corresponding row of an integer matrix $L$. Applying this decoder to the simulated execution recovers $\xi_i$ for every source vector $\xi$, proving the integer identity
\begin{equation}\label{eq:luj}
 LUJ=\Id_k.
\end{equation}
This identity is established on the zero-workspace input $J\xi$. The next construction removes any requirement to initialize the workspace of the final network.

Add logical banks $S=(S_i)$ at the sources and $Y=(Y_i)$ at the sinks. Write the full state as $(z,s,y)$ and define local maps
\[
 V(z,s,y)=(z+Js,s,y),\quad
 \Delta(z,s,y)=(z,s,y+Lz).
\]
The conjugate $\Theta=UVU^{-1}$ sends $(z,s,y)$ to $(z+UJs,s,y)$. Chronologically execute $\Theta$, then $\Delta$, then $\Theta^{-1}$, then $\Delta^{-1}$. The full calculation is
\begin{align*}
 (z,s,y)&\mapsto(z+UJs,s,y)\\
 &\mapsto(z+UJs,s,y+Lz+LUJs)\\
 &\mapsto(z,s,y+Lz+LUJs)\\
 &\mapsto(z,s,y+s).
\end{align*}
The final equality uses~\eqref{eq:luj}. Every initial workspace coordinate is restored, regardless of its value. This compute, copy, and uncompute pattern is the classical discipline of reversible computation~\cite{Bennett1973}.

Call this forward shear $P$. If its chronological gate sequence is $E_1,\ldots,E_\omega$, its matrix is $E_\omega\cdots E_1$. The chronological sequence $E_1^{-\mathsf T},\ldots,E_\omega^{-\mathsf T}$ therefore realizes $P^{-\mathsf T}$. Each new elementary gate exchanges control and target and negates its integer coefficient. This preserves locality and the underlying undirected cross edge. Thus $P^{-\mathsf T}$ is physically implementable and acts as $(z,s,y)\mapsto(z,s-y,y)$. The chronological sequence $P,P^{-\mathsf T},P$ gives
\begin{equation}\label{eq:signedswap}
 (z,s,y)\mapsto(z,s,y+s)
 \mapsto(z,-y,y+s)
 \mapsto(z,-y,s).
\end{equation}
Negating the first logical output recovers the exchanged input. This output operation is local. The signs in~\eqref{eq:signedswap} are needed over general groups.

The forward commutator shear executes four workspace words: two occurrences of $U$ and two of $U^{-1}$. Each word has $T$ cross phases, with at most one cross gate per physical edge in a phase. Inversion reverses the phase order and gate order; inverse transposition exchanges controls and targets and changes coefficients. Both operations preserve the disjoint register pairs within each phase and the one-use-per-physical-edge restriction. Thus the three shears have at most $12T$ cross phases and at most $12Tf$ cross gates on the $f$ cut edges.

We obtain a circuit $C$ on $N$ registers with $M$ elementary gates, a permutation $\pi$ of all registers, and output signs $\epsilon_j\in\{1,-1\}$ such that output $\pi(j)$, multiplied by $\epsilon_j$, equals input $j$. Workspace registers are fixed by the permutation. Writing $L_C$ for the number of cross phases and $g_F$ for the number of cross gates on $F$, we have
\begin{equation}\label{eq:phasebudget}
 L_C\leq12T,\qquad g_F\leq12Tf.
\end{equation}
The phases remain contiguous and register-disjoint. Local gates between them may share registers.

\subsection{The wire graph and its routing certificate}
Give each register its own input leaf and output leaf. Represent a gate by a vertex with two incoming and two outgoing ports, and join successive uses of a register by a wire; the wires of one register form its strand. Each wire is one physical undirected edge. An unused register has one direct input-to-output wire. At this stage the wire graph can have parallel edges and degree four. Its $N$ demands send input $j$ to output $\pi(j)$, including the workspace demands. Auxiliary host labels are annotations used to define the gates; they do not identify vertices or add edges in the wire graph.

Number the gates $1,\ldots,M$ and assign time labels $0$ to the inputs, $q$ to gate $q$, and $M+1$ to the outputs. If a wire joins an earlier label $\tau_1$ to a later label $\tau_2$, let $\alpha_e=\tau_2-\tau_1$. The wire spans of each register telescope. Summing over every register gives
\begin{equation}\label{eq:balance}
 \sum_e\alpha_e=N(M+1).
\end{equation}
At a cross gate whose auxiliary edge lies in $F$, mark the control input and control output ports. Let $\chi_e\in\{0,1,2\}$ count the marked endpoints of wire $e$. Each such gate contributes two marked incidences, so
\[
 \sum_e \chi_e=2g_F.
\]
A wire marked at both ends contributes two to this sum. Set
\begin{equation}\label{eq:ell}
 \ell_e=\alpha_e+2\chi_e,\qquad
 Q_0=\sum_e\ell_e=N(M+1)+4g_F.
\end{equation}

For any input-to-output walk, the signed temporal changes telescope to $M+1$. If $B_\tau$ is the total temporal span of its backward traversals, with multiplicity, its temporal length is
\begin{equation}\label{eq:backward}
 M+1+2B_\tau.
\end{equation}
Every walk therefore costs at least $M+1$, and one with a backward wire costs at least $M+3$.

Consider a wholly forward path for one of the $k$ designated demands, from logical source register $S_i$ to logical sink register $Y_i$. At a local gate the auxiliary host does not change. At a cross gate from control host $u$ to target host $v$, project its passage as follows:
\begin{equation}\label{eq:projection}
\begin{array}{c|c}
\text{input strand, output strand}&\text{auxiliary walk}\\\hline
\text{control, control}&u,v,u\\
\text{control, target}&u,v\\
\text{target, control}&v,u\\
\text{target, target}&v
\end{array}
\end{equation}
The out-and-back passage in the first row makes each traversal of an edge in $F$ correspond to one marked control port. Wire pieces join the projected passages at their common hosts.

A forward path meets at most one gate in each cross phase. Indeed, consider two consecutive visited gates in one phase. Every gate between them in the gate order lies in the same contiguous phase, so none is local, and the path leaves the first gate along a strand whose next use is the second gate. That register then occurs in both cross gates, contrary to phase disjointness. The auxiliary projection therefore has at most $2L_C\leq24T<b$ hops. It joins $s_i$ to $t_i$, so it must use $F$. The original path consequently pays at least one marked-port penalty of two.

We have proved lower bounds $h_j=M+1$ for every register demand, with $h_j=M+3$ for the $k$ designated source-bank demands. Therefore
\begin{equation}\label{eq:strictmetric}
 \begin{split}
 H_0:=\sum_jh_j&=N(M+1)+2k\\
 &>N(M+1)+48Tf\geq Q_0.
 \end{split}
\end{equation}
The strict inequality is~\eqref{eq:margins}, and the last step is~\eqref{eq:phasebudget}. Equation~\eqref{eq:lengthbound} certifies a routing ceiling strictly below one. The baseline equality~\eqref{eq:balance} explains why the workspace demands matter: omitting them would remove demand distance while leaving their wire budget in place.

\subsection{A connected subcubic scalar instance}\label{sec:split}
For the final degree bound, replace gate $q$ by two vertices $c_q,d_q$. The control input and output attach to $c_q$, the target input and output attach to $d_q$, and a new internal edge joins $c_q$ to $d_q$ (Fig.~\ref{fig:split}). The control vertex forwards $x$ along its strand and sends $x$ across the internal edge. The target vertex receives $x$ and forwards $y+a_qx$ along its own strand. Both vertices have degree three.

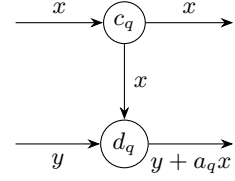
\begin{figure}[t]
\centering
\begin{tikzpicture}[x=0.9cm,y=0.8cm,>=Stealth, every node/.style={font=\small}]
\node[circle,draw,inner sep=2pt] (c) at (0,1) {$c_q$};
\node[circle,draw,inner sep=2pt] (d) at (0,-1) {$d_q$};
\draw[->] (-1.6,1)--node[above]{$x$}(c);
\draw[->] (c)--node[above]{$x$}(1.6,1);
\draw[->] (-1.6,-1)--node[below]{$y$}(d);
\draw[->] (d)--node[below]{$y+a_qx$}(1.6,-1);
\draw[->] (c)--node[right]{$x$}(d);
\end{tikzpicture}
\caption{A split elementary gate. Arrows indicate the scalar code's fixed direction; all five physical edges are undirected with shared capacity one. The two gate vertices have degree three.}
\label{fig:split}
\end{figure}

Splitting yields a simple graph directly. Two strand wires with the same endpoints would have the same register and the same earlier occurrence, hence be the same wire. A strand wire cannot coincide with an internal edge, whose endpoints belong to one gate. Distinct internal edges belong to distinct gates. Inputs and outputs remain distinct leaves.

Give strand wires their old lengths $\ell_e$, and give internal edges length zero. Contracting the internal edges projects every split-graph walk to an old wire-graph walk with the same inherited length, after deleting stationary steps. Thus~\eqref{eq:strictmetric} remains valid, with the same total budget and demand lower bounds.

There is also an explicit scalar schedule using each split-graph edge once. For gate $q\in\{1,\ldots,M\}$, send its internal edge at time $2q-1$, and its outgoing strand wires at time $2q$. Wires leaving inputs are sent at time zero. Every incoming strand of gate $q$ arrives strictly before $2q-1$. Thus the control symbol is locally available for the internal send, and the target has both symbols before its outgoing send. The entire scalar code has duration $1+2M$. Applying the local signs at the output gives exact decoding for every register input.

The split graph has the right local degree and code properties. Its certificate still has zero-length internal edges, unequal demand lower bounds, and possibly several components. We remove these features in a way that turns the weighted certificate into an ordinary hop-distance bound.

\begin{lemma}[Positive connected certificate]\label{lem:connected}
There is a finite connected simple subcubic graph with distinct leaf terminals, a scalar integer code using each edge exactly once, positive integer edge lengths $\ell_e$, and nonnegative integer own-pair lower bounds $h_i$, such that
\begin{equation}\label{eq:weightedsurplus}
 \sum_e\ell_e<\sum_i h_i.
\end{equation}
\end{lemma}
\begin{proof}
Let $m_*$ be the number of edges of the split graph and $a=m_*+1$. Replace every old length by $a\ell_e+1$ and every demand bound by $ah_i$. All path bounds persist because the added unit lengths are nonnegative. Since the old surplus $H_0-Q_0$ is a positive integer, the new surplus is
\[
 a(H_0-Q_0)-m_*\geq1.
\]
Every new edge length is positive.

Exact decoding forces each demanded source and sink to belong to the same component. To see this, vary one source coordinate while fixing all others. By strict-causality induction, events in a component depend only on inputs whose sources lie there. A sink in another component could not recover the varied coordinate.

Partition both the edge budget and the demand sum by components. Their total difference is positive, so some component has a positive difference. It contains a demand, since its edge budget is nonnegative. Restrict the code to that component and retain exactly its sessions. Every retained event has its sender in the component; any local source used by it is there, and any predecessor arrives along an edge whose other endpoint is also there. Thus its local formula uses only retained sources and retained events. Specializing outside source coordinates to zero in the original row identities proves the retained decoder identities. Locality, strict causality, exactness, and one use per retained edge all persist. Simplicity, degree bounds, and terminal leaves are inherited.
\end{proof}

\begin{lemma}[Hop normalization]\label{lem:normalize}
The graph in Lemma~\ref{lem:connected} can be transformed into a scalar instance with $\rho<1$.
\end{lemma}
\begin{proof}
Let $k$ be the retained number of sessions and put $h_*=\sum_i h_i>0$. For each session $i$, increase the length of its source leaf edge by $h_*-h_i\geq0$. Every own-pair walk uses its source leaf edge, so its new weighted length is at least $h_*$. The total padded edge budget is
\begin{equation}\label{eq:padbudget}
 Q'=\sum_e\ell_e+\sum_i(h_*-h_i)<kh_*.
\end{equation}
Call the resulting positive integer lengths $\ell'_e$.

Replace edge $e$ by a private path of $2\ell'_e$ unit edges. The resulting graph is connected, simple, and subcubic, and all old terminals remain distinct leaves. Each new edge has a nonterminal endpoint, because every replacement path has at least two edges and hence an interior vertex. Every path between old terminals traverses an entire private replacement whenever it enters its interior; erasing excursions and collapsing these paths gives an old walk. Hence every own-pair distance is at least $b_0=2h_*$. Its edge count is $m_0=2Q'$, so
\begin{equation}\label{eq:seedceiling}
 \rho_0:=\frac{m_0}{kb_0}=\frac{Q'}{kh_*}<1.
\end{equation}

To relay the scalar code, let $s_*=1+\sum_e2\ell'_e$. Replace an old event at time $t$ by successive sends on its private path at times
\[
 s_*t,\ s_*t+1,\ \ldots,\ s_*t+2\ell'_e-1.
\]
An old predecessor at $t'<t$ reaches its endpoint by time $s_*(t'+1)-1<s_*t$. Thus every old local expression is still available when required. Each new edge is used once. The output identities and local signs are unchanged. With the resulting finite duration $T_0$, all fields of Definition~\ref{def:scalar} are satisfied.
\end{proof}

\begin{theorem}[Alphabet-independent seed]\label{thm:seed}
There exists a scalar instance $\mathcal S_0$ with $0<\rho_0<1$. Its graph, scalar schedule, and integer coefficients are selected without choosing a field or group.
\end{theorem}
\begin{proof}
Use Lemma~\ref{lem:input} with the parameters in~\eqref{eq:seednumbers}. The reversible construction, the split schedule, and~\eqref{eq:strictmetric} are the inputs to the proof of Lemma~\ref{lem:connected}. Apply Lemma~\ref{lem:normalize}. Every step is a finite graph construction or an integer row identity. Positivity of $m_0,k,b_0$ gives $\rho_0>0$.
\end{proof}

\subsection{A finite block above the routing ceiling}\label{sec:finite-pipeline}
The scalar instance has a strict routing ceiling, but its scalar execution may have a long duration. To prove Theorem~\ref{thm:finite}, we overlap independent executions. The one-use edge property makes a one-round spacing possible.

\begin{lemma}[Unit-start pipeline]\label{lem:pipeline}
For every scalar instance and every $B>0$, an exact integer block protocol delivers $B$ symbols per session in $B+T$ rounds.
\end{lemma}
\begin{proof}
Index the independent scalar blocks by $a=0,\ldots,B-1$. If an event has scalar time $\tau_e$, send its symbol for block $a$ at time $a+\tau_e$. Each physical edge has one scalar event, so distinct blocks use that edge at distinct times. Within one block, every predecessor remains strictly earlier. Induction on scalar time proves each payload and decoder identity. The final time is less than $B+T$; unused final slots cause no problem. The integer local expressions use only the coordinates of their own block.
\end{proof}

\begin{proof}[Proof of Theorem~\ref{thm:finite}]
Take the scalar instance $\mathcal S_0$ of Theorem~\ref{thm:seed}, with parameters $k_0,m_0,b_0,T_0$. Its strict certificate says $m_0<k_0b_0$. Choose
\begin{equation}\label{eq:finiteblock}
 B_0=m_0T_0+1,\qquad D_0=B_0+T_0.
\end{equation}
The positive integer $\delta=k_0b_0-m_0$ satisfies
\[
 B_0\delta\geq B_0>m_0T_0.
\]
Rearranging this strict inequality gives
\begin{equation}\label{eq:finitestrict}
 \frac{B_0}{D_0}>\frac{m_0}{k_0b_0}
 =\rho_0\geq R_\rt(\mathcal S_0).
\end{equation}
The last inequality is~\eqref{eq:hopceiling}. Lemma~\ref{lem:pipeline} supplies an exact integer block protocol of duration $D_0$ for these $B_0$ symbols. Interpret every coefficient in $\mathbb F_2$. All local maps and decoders are binary linear, and every shared-capacity and strict-causality condition is unchanged. The scalar-instance definition supplies the graph and terminal properties. Thus this finite binary protocol proves~\eqref{eq:finitecounterexample}.
\end{proof}

\section{Alphabet-Uniform Coding and Capacity Regions}\label{sec:extensions}
Theorem~\ref{thm:finite} used the binary interpretation of a code defined over the integers. No part of the graph, schedule, or local row identities depended on this choice. Interpreting the same identities in any nontrivial finite abelian group gives an exact code on the same graph. In a finite field, integer coefficients become field elements and each local function is linear over that field. In particular, the finite block length~\eqref{eq:finiteblock} separates coding from routing over all these alphabets simultaneously.

The integer construction gives more than one separated rate. Lemma~\ref{lem:pipeline} supplies common rates $B/(B+T)$ approaching one on every scalar instance. We determine the full closure of its vector-rate region by combining this pipeline with the capacity bound at a source leaf. These conclusions will also apply to every amplified and subdivided scalar instance.

\begin{lemma}[Leaf bound]\label{lem:leaf}
If $s_i\neq t_i$ and the source $s_i$ has one incident edge, every exact deterministic fixed-schedule protocol over a finite alphabet $A$, $|A|>1$, satisfies $B_i\leq D$.
\end{lemma}
\begin{proof}
Fix all messages except $X_i$. Record the complete sequence of symbols sent on the source leaf edge, in both directions, padding its unused rounds with a fixed symbol. There are at most $|A|^D$ records. Equal records yield identical information outside the source leaf: induct on rounds, using the identical incoming edge transcript and the fixed other messages. The sink lies outside the leaf and decodes $X_i$ exactly, so distinct values of $X_i$ must have distinct records. Hence $|A|^{B_i}\leq|A|^D$, which implies $B_i\leq D$. This proof allows arbitrary deterministic local functions.
\end{proof}

\begin{proposition}\label{prop:capacity}
Every scalar instance has common-rate supremum one over every nontrivial finite abelian group, and field-linear common-rate supremum one over every finite field. Its symbol and field-linear vector regions have closure $[0,1]^k$.
\end{proposition}
\begin{proof}
Lemma~\ref{lem:pipeline} gives rates $B/(B+T)\to1$, simultaneously under every group or field interpretation. Lemma~\ref{lem:leaf} bounds every coordinate by one, even for nonlinear symbol functions.

For a target vector $r\in[0,1]^k$, take $B_i=\lfloor r_iB\rfloor$ and use the $B$-symbol pipeline after setting the unused coordinates of session $i$ to zero. This is a restriction of the integer local formulas and their decoding identities. It is valid for unequal message lengths and gives
\begin{equation}\label{eq:vectorerror}
 0\leq r_i-\frac{\lfloor r_iB\rfloor}{B+T}
 <\frac{T+1}{B+T}.
\end{equation}
Thus every point of the cube is a limit of achievable vectors. The cube is closed, and the leaf bound contains both rate regions in it, proving equality of the closures.
\end{proof}

The finite-block regions themselves can be identified exactly. Rate one is never attained by a finite block, and every rational vector below one is attained. For an infinite field $\F$, the region $\mathcal R_\lin(G;\F)$ is defined by the same conditions as in Section~\ref{sec:model}, with rates counted in field symbols per round.

\begin{proposition}[Exact finite-block region]\label{prop:exactregion}
For every scalar instance, the set of rate vectors of exact deterministic fixed-schedule codes over a nontrivial finite abelian group, and the set of rate vectors of field-linear codes over any field, finite or infinite, both equal the set of rational vectors in $[0,1)^k$.
\end{proposition}
\begin{proof}
Fix a session $i$, its source leaf edge $e$, and a code of duration $D$ in which session $i$ has $B_i$ coordinates. The other endpoint of $e$ is a nonterminal by Definition~\ref{def:scalar}, so it is not $t_i$: a symbol sent on $e$ in round $D-1$ reaches neither $t_i$ nor any later event, and $t_i$ receives no symbol directly from $s_i$. Repeat the argument of Lemma~\ref{lem:leaf} with the transcript of $e$ restricted to rounds $0,\ldots,D-2$. Every symbol sent by a vertex other than $s_i$ depends only on fixed messages, on symbols from such vertices in earlier rounds, and on symbols of $e$ in rounds at most $D-2$. By induction on rounds, equal shortened transcripts therefore give the same symbols sent by every vertex other than $s_i$, hence the same decoded message. Thus
\[
 |A|^{B_i}\leq|A|^{D-1},\qquad\text{so}\qquad B_i\leq D-1\quad\text{and}\quad\frac{B_i}{D}<1.
\] For a field-linear code with the other messages set to zero, every transmitted symbol is, by the same induction, an $\F$-linear function of $X_i$, so the shortened transcript is a linear and injective map $\F^{B_i}\to\F^{D-1}$, and a comparison of dimensions gives $B_i\leq D-1$ without any finiteness of the field. All rates are quotients of integers.

Conversely, let $r\in[0,1)^k$ be rational, with common denominator $d$ and numerators $a_i<d$. Put
\[
 D=d(T+1),\qquad B=D-T,\qquad B_i=a_i(T+1).
\]
Since $a_i+1\leq d$, we have $B_i\leq B$. Give session $i$ the first $B_i$ coordinates of the $B$-symbol pipeline of Lemma~\ref{lem:pipeline} and set the remaining coordinates to zero. The duration is $D$, and the rate of session $i$ is
\[
 \frac{B_i}{D}=\frac{a_i(T+1)}{d(T+1)}=\frac{a_i}{d}.
\] The protocol is the integer pipeline, so it is exact over every group and linear over every field.
\end{proof}

\section{Polylogarithmic Gap Amplification}\label{sec:amplification}
We amplify the scalar seed by composing copies of a network with one another. The construction follows the tensor-product principle of BGS~\cite{BGS2017}; here we retain the integer local code, degree bound, and leaf terminals at each step. The complementary colored cover of Definition~\ref{def:cover} specifies the attachments. The quantitative bound in Lemma~\ref{lem:cover} determines the final gap as a function of network size.

\subsection{The degree-preserving tensor construction}\label{sec:tensor}
For $i=1,2$, let $\mathcal S_i$ be a scalar instance with parameters $(k_i,m_i,b_i,T_i)$, vertex count $n_i$, and certified ceiling $\rho_i$. Use $\mathcal S_1$ as the outer instance and $\mathcal S_2$ as the inner instance. Choose the cover in Lemma~\ref{lem:cover} with
\[
 d=m_1,\qquad k=k_2,\qquad g=2b_1b_2.
\]
Each outer physical edge has its fixed direction from the unique scalar-code event. Only that direction is used in the attachment construction; the resulting graph is still undirected.

Take one outer vertex set for each left cover vertex $l$ and one inner graph for each right cover vertex $z$. An incidence of colors $(e,i)$ attaches inner source $s_i$ to the tail of outer edge $e$ in copy $l$, and inner sink $t_i$ to its head (Fig.~\ref{fig:tensor}). Remove all outer edges and retain all inner physical edges. Keep the copied outer terminal pairs as the new demands.

Every outer copy carries a type $i\in[k_2]$ and is served by session $i$ of each inner copy it touches; every inner copy carries a type $e\in[m_1]$ and replaces the outer edge $e$ in each of the $k_2$ outer copies it touches, one per session.

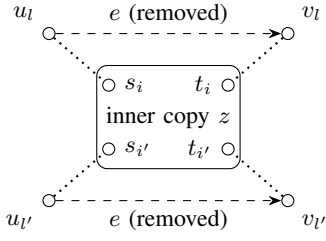
\begin{figure}[t]
\centering
\begin{tikzpicture}[x=0.79cm,y=0.65cm,>=Stealth,
  every node/.style={font=\small},term/.style={circle,draw,inner sep=1.6pt}]
\node[term,label=above left:{$u_l$}] (u1) at (0,1.7) {};
\node[term,label=above right:{$v_l$}] (v1) at (4,1.7) {};
\node[term,label=below left:{$u_{l'}$}] (u2) at (0,-1.7) {};
\node[term,label=below right:{$v_{l'}$}] (v2) at (4,-1.7) {};
\draw[->,dashed] (u1)--node[above]{$e$ (removed)}(v1);
\draw[->,dashed] (u2)--node[below]{$e$ (removed)}(v2);
\draw[rounded corners] (0.8,-1.05) rectangle (3.2,1.05);
\node[term,label=right:{$s_i$}] (s1) at (1,0.65) {};
\node[term,label=left:{$t_i$}] (t1) at (3,0.65) {};
\node[term,label=right:{$s_{i'}$}] (s2) at (1,-0.65) {};
\node[term,label=left:{$t_{i'}$}] (t2) at (3,-0.65) {};
\node[align=center] at (2,0) {inner copy $z$};
\draw[dotted,thick] (u1)--(s1);
\draw[dotted,thick] (v1)--(t1);
\draw[dotted,thick] (u2)--(s2);
\draw[dotted,thick] (v2)--(t2);
\end{tikzpicture}
\caption{Two incidences attached to one inner copy of right type $e$. Dotted connectors identify terminal vertices with outer endpoints; they are contracted and consume no physical edges. The old dashed edges are removed. The inner graph's physical edges, not shown, are retained. All sessions assigned to this inner copy carry the same edge event in different outer copies.}
\label{fig:tensor}
\end{figure}

Every inner terminal participates in exactly one incidence. Consequently the identifications consist of stars centered at outer vertices, together with isolated inner nonterminals. Distinct outer vertices never merge. Within one inner copy, different sessions attach to different outer copies because the incidence graph is simple; the two endpoints of one attachment are distinct because the outer edge has no loop. Hence every inner-copy vertex map is injective.

Writing $\phi$ for the common fiber size in~\eqref{eq:covercounts}, the exact counts are
\begin{equation}\label{eq:tensorcounts}
 k_+=k_1k_2\phi,\qquad m_+=m_1m_2\phi,
 \qquad b_+=b_1b_2.
\end{equation}
The distance assertion is proved below. The vertex count is
\begin{equation}\label{eq:tensorvertices}
 n_+=|V_L|n_1+|V_R|(n_2-2k_2),
\end{equation}
since the remaining vertices are the outer vertices and the private inner nonterminals.

The new graph is simple. Each retained edge has a private inner nonterminal endpoint, which determines its inner copy. Two edges with the same unordered endpoints therefore come from the same copy, where simplicity and injectivity distinguish them. Connectivity follows because an inner copy connects the two ends of every outer edge it replaces, and the connected incidence graph joins all copies.

At an outer vertex, each old incident edge is replaced by exactly one incident inner leaf edge. Its degree is unchanged. An inner nonterminal keeps its old degree. Thus maximum degree three and the distinct leaf terminals persist. Every retained edge still has a nonterminal endpoint in the new graph.

\subsection{The composed code and its distance bound}
All incidences at a fixed right vertex have the same outer-edge color $e$. The unique event using that edge has one fixed outer time $\theta_e$ across the outer copies. Therefore the entire inner copy can run once, in the interval beginning at $T_2\theta_e$, with its sessions carrying the appropriate outer payloads. An inner event of time $\tau$ is sent at time $T_2\theta_e+\tau$.

When outer round $\theta_e$ begins, each source of an inner session is the identified outer sender and knows that payload from its local outer history. The inner code decodes every assigned payload by the end of the interval. Its identities remain valid when these inputs are correlated integer forms in the original messages. Strict inner predecessors remain earlier within the interval, and outer predecessors finish in earlier intervals. Induction on outer time gives the original outer decoder at every new sink.

Each inner copy executes once. Since its scalar code uses each of its physical edges once, the complete tensor code uses every new edge once. Different copies have different physical edges. The result is an exact integer scalar code of duration $T_1T_2$, with fixed reference directions on the retained edges.

\begin{lemma}\label{lem:tensordistance}
Every undirected own-pair walk in the tensor graph has at least $b_1b_2$ edges.
\end{lemma}
\begin{proof}
Expand the identifications into zero-length connectors between outer vertices and inner terminals. Lift a candidate walk and erase loops to obtain a simple expanded path without increasing its physical-edge count $p$. Record the sequence of copies it visits; this closed walk in the cover, based at the outer copy containing both endpoints, is its incidence trace, and the set of cover edges it uses is the trace support. Each passage through an inner copy requires a physical edge between its two connector crossings, so the number $\kappa$ of crossings satisfies $\kappa\leq2p$.

If the trace support contains a cycle, its length is at least $2b_1b_2$ by the cover girth. Thus $\kappa\geq2b_1b_2$ and $p\geq b_1b_2$.

If the trace support is a tree, root it at the endpoint outer copy. Every used incidence edge is crossed twice, once at each of its two distinct connectors. The path therefore consists of nested excursions. An inner excursion either travels directly between its assigned own-pair terminals within that inner copy, costing at least $b_2$, or contains a deeper inner excursion that already costs at least $b_2$. Induction on excursion depth establishes this bound in all cases. At the root, contracting each maximal inner excursion to its outer edge gives an own-pair walk in the outer graph. It has at least $b_1$ edges. These maximal excursions are disjoint path portions, each of cost at least $b_2$, giving the claimed product. Appendix~\ref{app:excursions} proves the nesting and induction in detail.
\end{proof}

\begin{proposition}[Scalar tensor]\label{prop:tensor}
The construction gives a scalar instance with duration $T_1T_2$, the counts in~\eqref{eq:tensorcounts}, and certified ceiling
\begin{equation}\label{eq:ceilingproduct}
 \rho_+=\rho_1\rho_2.
\end{equation}
\end{proposition}
\begin{proof}
The topology and scalar code were established above. Lemma~\ref{lem:tensordistance} supplies the distance field. Substitution of~\eqref{eq:tensorcounts} in $m_+/(k_+b_+)$ cancels $\phi$ and gives~\eqref{eq:ceilingproduct}.
\end{proof}

\subsection{Iteration and the logarithmic gap}\label{sec:growth}
Starting with the seed $\mathcal S_0$ of Theorem~\ref{thm:seed}, define $\mathcal S_{j+1}$ by tensoring $\mathcal S_j$ with itself. Let its parameters be $k_j,m_j,b_j,T_j$ and its vertex count be $v_j$. Equation~\eqref{eq:ceilingproduct} and the distance recurrence give
\begin{equation}\label{eq:recurrences}
 \rho_j=\rho_0^{2^j},\qquad b_j=b_0^{2^j},
 \qquad T_j=T_0^{2^j}.
\end{equation}
In particular every routing rate is at most $\rho_0^{2^j}$, while Proposition~\ref{prop:capacity} applies to every $\mathcal S_j$.

It remains to bound the vertex count. Since every terminal is a leaf and every edge has a nonterminal endpoint, distinct terminals have distinct incident edges; hence $2k_j\leq m_j$. The cover bound~\eqref{eq:coverbound} therefore gives
\begin{equation}\label{eq:mbound}
 m_{j+1}\leq m_j^2\left(2(m_j+1)^{2b_j^2}+1\right).
\end{equation}
For positive $m_j,b_j$, it follows that
\begin{equation}\label{eq:mbound2}
 3m_{j+1}\leq(3m_j)^{4b_j^2}.
\end{equation}
Indeed the left side is at most $9m_j^2(m_j+1)^{2b_j^2}$, bounded by $(3m_j)^{2+2b_j^2}$, and $2+2b_j^2\leq4b_j^2$.

Put
\begin{equation}\label{eq:constants}
 \begin{gathered}
 \Gamma=\max\{3m_0,b_0\}+1,\qquad \Lambda=4\Gamma^2,\\
 c=\frac{\log(1/\rho_0)}{\log\Lambda}>0.
 \end{gathered}
\end{equation}
We claim
\begin{equation}\label{eq:sizebound}
 3m_j\leq2^{\Lambda^{2^j}}.
\end{equation}
For $j=0$, $3m_0\leq\Lambda\leq2^\Lambda$. Further,
\[
 4b_j^2=4b_0^{2\cdot2^j}\leq(4b_0^2)^{2^j}
 \leq\Lambda^{2^j}.
\]
Using this inequality in~\eqref{eq:mbound2}, the induction step raises $2^{\Lambda^{2^j}}$ to a power at most $\Lambda^{2^j}$, giving $2^{\Lambda^{2^{j+1}}}$.

Every connected nontrivial graph has $v_j\leq m_j+1\leq2m_j$, so~\eqref{eq:sizebound} also bounds $v_j$. Because $\log2<1$,
\begin{equation}\label{eq:logpower}
 (\log v_j)^c\leq(\Lambda^{2^j})^c
 =(1/\rho_0)^{2^j}.
\end{equation}
Combining with~\eqref{eq:recurrences} proves
\begin{equation}\label{eq:pre-girth-gap}
 R_\rt(\mathcal S_j)(\log v_j)^c\leq1.
\end{equation}

The vertex counts are unbounded. The inequality $2k_0\leq m_0<k_0b_0$ implies $b_0>2$. A shortest own-pair path has at most $v_j-1$ edges and at least $b_j$ edges, so $v_j>b_0^{2^j}>j$. This also confirms that all logarithms used above are positive.

\section{Bipartite Families with Polynomial Girth}\label{sec:girth}
The incidence girth in the tensor construction controls source--sink distances. It does not bound every physical cycle, because cycles inside an inner copy remain. We now subdivide the amplified instances. Even replacement lengths give bipartiteness, and a suitable length raises the physical girth to a power of the new vertex count.

\subsection{Subdividing the physical edges}
Long replacement paths increase propagation delay, but a scalar protocol still uses each new edge only once. They therefore preserve the limiting coding rate and let us control the lengths of physical cycles.

\begin{lemma}[Uniform subdivision]\label{lem:subdivision}
Let a scalar instance have $n$ vertices and parameters $(k,m,b,T)$. Replacing every edge by a private path of $\lambda\geq2$ edges gives a scalar instance with
\begin{equation}\label{eq:subcounts}
 \begin{gathered}
 n'=n+(\lambda-1)m,\qquad m'=\lambda m,\\
 b'=\lambda b,\qquad T'=T(m\lambda+1),
 \end{gathered}
\end{equation}
and the same ceiling $m'/(kb')=m/(kb)$. Every cycle has at least $\lambda$ edges. If $\lambda$ is even, the new graph is bipartite.
\end{lemma}
\begin{proof}
The relay schedule used in Lemma~\ref{lem:normalize}, now with common replacement length $\lambda$ and spacing $m\lambda+1$, gives an exact scalar code using every new edge once and duration $T(m\lambda+1)$. Every terminal-to-terminal simple path contracts to an old path and traverses full replacement paths, giving the distance factor $\lambda$. The counts are immediate; simplicity, connectivity, degree bounds, and leaf terminals persist, and every new edge has a private interior endpoint.

A cycle cannot stay inside a replacement path, which is a path with private degree-two internal vertices. Once a cycle uses one of its internal vertices, it must continue through the path to an old endpoint in both directions; it therefore includes a whole replacement path and has length at least $\lambda$. For even $\lambda$, color every old vertex with color zero and every subdivision vertex by the parity of its position along its replacement path. Both endpoints receive color zero consistently, and every edge changes color.
\end{proof}

\subsection{The simultaneous family theorem}
The full result combines the finite counterexample, its alphabet-independent scalar representation, the tensor growth estimate, and the subdivision bound.

\begin{theorem}\label{thm:main}
There is a constant $c>0$ with the following property. For every $\beta\in(0,1)$ there are a constant $K_\beta>0$ and networks $G_0,G_1,\ldots$, with $n_j=|V(G_j)|>j$ and $k_j>0$ sessions, such that:
\begin{enumerate}
\item $G_j$ is connected, simple, undirected, and bipartite. Every physical edge has shared capacity one, every vertex has degree at most three, and the $2k_j$ terminals are pairwise distinct leaves.
\item Every cycle of $G_j$ has at least $n_j^\beta$ edges, and
\begin{equation}\label{eq:main-routing}
 R_\rt(G_j)(\log n_j)^c\leq K_\beta.
\end{equation}
\item For some positive integer $T'_j$, every positive integer $B$ has an exact causal integer-coefficient block protocol of duration $B+T'_j$, delivering $B$ symbols per session. The graphs, schedules, and integer coefficients are chosen independently of the alphabet.
\item For every nontrivial finite abelian group $A$, the supremum of deterministic symbol-code common rates on $G_j$ is one and
\[
 \cl\mathcal R_\sym(G_j;A)=[0,1]^{k_j}.
\]
For every finite field $\F$, the supremum of field-linear common rates is one and
\[
 \cl\mathcal R_\lin(G_j;\F)=[0,1]^{k_j}.
\]
\item For some positive integer $B$, the protocol of item~3 has common rate $B/(B+T'_j)>R_\rt(G_j)$.
\end{enumerate}
\end{theorem}

To choose the subdivision length, consider a scalar iterate with $n\geq2$ vertices and $m$ edges. Write $n'$ for its vertex count after subdivision, as in Lemma~\ref{lem:subdivision}. Fix $\beta\in(0,1)$ and put
\begin{equation}\label{eq:subparams}
 p=\frac1{1-\beta},\qquad a=\frac\beta{1-\beta},
 \qquad \lambda=2\left\lceil(2n)^a\right\rceil.
\end{equation}
This is an even integer at least two. The subcubic handshake bound is $2m\leq3n$, and hence
\begin{equation}\label{eq:subsize}
 n'=n+(\lambda-1)m\leq2n\lambda.
\end{equation}
As $\lambda\geq(2n)^a$, we have $(2n)^\beta\leq \lambda^{1-\beta}$. Therefore
\begin{equation}\label{eq:girthpower}
 (n')^\beta\leq(2n\lambda)^\beta\leq \lambda.
\end{equation}
Lemma~\ref{lem:subdivision} now gives girth at least $(n')^\beta$.

The choice of $\lambda$ also preserves the logarithmic exponent. Since $(2n)^a\geq1$, the ceiling in~\eqref{eq:subparams} yields $\lambda\leq4(2n)^a$, and thus
\begin{equation}\label{eq:polynomialsize}
 n'\leq \eta_\beta n^p,\qquad \eta_\beta=8\cdot2^a.
\end{equation}
For $n\geq2$, taking logarithms gives
\begin{equation}\label{eq:logtransfer}
 \log n'\leq \theta_\beta\log n,\qquad
 \theta_\beta=p+\frac{\log \eta_\beta}{\log2}.
\end{equation}
For the subdivision $G_j$ of $\mathcal S_j$, the unchanged certificate gives $R_\rt(G_j)\leq\rho_j$. Equations~\eqref{eq:recurrences} and~\eqref{eq:logpower} give $\rho_j(\log v_j)^c\leq1$. Multiplying this inequality by $\theta_\beta^c$ and using~\eqref{eq:logtransfer} proves
\[
 R_\rt(G_j)(\log |V(G_j)|)^c\leq \theta_\beta^c.
\]
Set $K_\beta=\theta_\beta^c>0$.

\begin{proof}[Proof of Theorem~\ref{thm:main}]
Choose the integer seed once by Theorem~\ref{thm:seed}, and define $c$ by~\eqref{eq:constants}. Construct the scalar iterates $\mathcal S_j$ by Proposition~\ref{prop:tensor}. For each fixed $\beta$, subdivide $\mathcal S_j$ by the even length~\eqref{eq:subparams}, with $n=v_j$, to obtain $G_j$. Lemma~\ref{lem:subdivision} proves its scalar-instance properties and bipartiteness. Equations~\eqref{eq:girthpower} and~\eqref{eq:logtransfer} prove the physical girth and routing bound. Its vertex count is at least $v_j>j$.

Apply Lemma~\ref{lem:pipeline} to each final scalar code, whose duration $T'_j=T_j(m_j\lambda+1)$ is given by~\eqref{eq:subcounts}. This yields the $B+T'_j$ protocols in item~3. Proposition~\ref{prop:capacity} proves item~4. Since every graph, cover, relay schedule, and integer coefficient was selected before an alphabet was interpreted, all of these conclusions hold on the same family for every group and field.

For each $j$, the retained routing ceiling is $\rho_j<1$. Choosing $B>\rho_jT'_j/(1-\rho_j)$ gives $B/(B+T'_j)>\rho_j\geq R_\rt(G_j)$. This proves item~5.
\end{proof}

A longer subdivision gives a single family whose girth is nearly linear in the number of vertices, at the cost of halving the exponent.

\begin{corollary}[A single family with near-linear girth]\label{cor:nearlinear}
There are constants $c'>0$ and $K>0$ and networks $G'_0,G'_1,\ldots$ with $n'_j=|V(G'_j)|>j$ that satisfy items~1, 3, 4, and~5 of Theorem~\ref{thm:main}, such that every cycle of $G'_j$ has at least $n'_j/(2\exp\sqrt{\log n'_j})$ edges and $R_\rt(G'_j)(\log n'_j)^{c'}\leq K$. In particular, for every fixed $\beta<1$, every cycle of $G'_j$ has at least $(n'_j)^\beta$ edges for all sufficiently large $j$.
\end{corollary}
\begin{proof}
Put $t=\log v_j$, and recall that $v_j\geq2$ and $m_j\geq1$ (Section~\ref{sec:growth}). Subdivide every edge of $\mathcal S_j$ into a path of
\begin{equation}\label{eq:nearlinearlength}
 \lambda=2\left\lceil e^{t^2}\right\rceil
\end{equation}
unit edges. By Lemma~\ref{lem:subdivision}, the result satisfies items~1, 3, 4, and~5 exactly as in the proof of Theorem~\ref{thm:main}, every cycle has at least $\lambda$ edges, and the routing ceiling is still $\rho_j$.

It remains to compare $\lambda$ and $\rho_j$ with the new vertex count $n'_j=v_j+(\lambda-1)m_j$. Since $m_j\geq1$ and $2m_j\leq3v_j$,
\begin{equation}\label{eq:nearlinearsize}
 \lambda\leq n'_j\leq2v_j\lambda,
\end{equation}
and the choice~\eqref{eq:nearlinearlength} gives $e^{t^2}\leq\lambda\leq4e^{t^2}$. The lower bounds give $n'_j\geq e^{t^2}$, that is, $t\leq\sqrt{\log n'_j}$, and therefore
\begin{equation}\label{eq:nearlineargirth}
 \lambda\geq\frac{n'_j}{2v_j}\geq\frac{n'_j}{2\exp\sqrt{\log n'_j}},
\end{equation}
which is the girth bound. The upper bounds give $n'_j\leq8v_je^{t^2}$, so
\begin{equation}\label{eq:nearlinearlog}
 \log n'_j\leq\log8+t+t^2\leq9t^2,
\end{equation}
using $t\geq\log2$. Put $c'=c/2$. Raising~\eqref{eq:nearlinearlog} to the power $c'$ and using $\rho_jt^{c}\leq1$, which is~\eqref{eq:recurrences} together with~\eqref{eq:logpower}, gives
\begin{equation}\label{eq:nearlinearrouting}
 \begin{split}
 R_\rt(G'_j)(\log n'_j)^{c'}
 &\leq\rho_j(\log n'_j)^{c'}\\
 &\leq9^{c/2}\rho_jt^{c}
 \leq9^{c/2}=:K.
 \end{split}
\end{equation}
Finally, the bound~\eqref{eq:nearlineargirth} is at least $(n'_j)^\beta$ as soon as $(1-\beta)\log n'_j\geq\log2+\sqrt{\log n'_j}$, which holds for all large $j$ because $n'_j\geq v_j>j$.
\end{proof}

The threshold in the last statement cannot be removed. A cycle avoids the $2k_j$ leaf terminals, so it has at most $n'_j-2k_j$ edges, and no fixed network with a cycle has girth at least $n^\beta$ for every $\beta<1$.

\section{Concluding Remarks}
This paper settles a question that has stood since 2004: in undirected networks with multiple unicast sessions, network coding can beat fractional routing. The counterexample is a finite network built from a reversible circuit in which every register, workspace included, carries its own unicast demand. The total temporal length of the wires then equals the total baseline distance of the demands exactly, and the designated demands, whose forward routes cannot avoid the cut set and whose other routes pay for backward steps, tip the balance against routing; the code uses each wire once and pipelines toward rate one.

We were struck by how far this one idea carries. The same integer instructions work over every finite field and every nontrivial finite abelian group. The routing ceiling is squared by each tensor step and preserved by subdivision, which together give, for every $0<\beta<1$, unbounded families of subcubic bipartite networks with girth at least $n^\beta$ and a polylogarithmic coding advantage, and a single family with girth $n^{1-o(1)}$. The finite counterexample and the family theorem are formalized in Lean, relying on its standard axioms only.

Much remains open: the true advantage of the finite instance, whose routing rate our certificate bounds only from above; whether the $O(\log k)$ upper bound on the coding advantage is attained; and whether a small counterexample exists. Coding, it turns out, does help in undirected networks; how much it can help remains to be seen.

\appendices
\section{Finite Colored Covers}\label{app:covers}
The covers of Definition~\ref{def:cover} are finite incidence graphs used in both the auxiliary recursion and the scalar tensor construction. The proof below establishes their existence and size bound independently of either application.

\begin{proof}[Proof of Lemma~\ref{lem:cover}]
Write $\mu=\max\{d,k\}$. Use left vertices $(i,x)\in[k]\times[\nu]$ and right vertices $(a,y)\in[d]\times[\nu]$. For each pair $(i,a)$ choose a permutation $\sigma_{i,a}$ of $[\nu]$, and join $(i,x)$ to $(a,\sigma_{i,a}(x))$. Every choice gives a simple colored cover of maximum degree $\mu$.

The space of choices is finite. Among them maximize the girth truncated at $g+1$; among choices attaining that maximum, minimize the number of cycles of minimum length. Suppose the selected graph has shortest cycle length $r\leq g$, and choose an edge $uv$ on such a cycle. It belongs to one $(i,a)$ matching.

A radius-$g$ ball in a maximum-degree-$\mu$ graph has at most $(\mu+1)^g$ vertices. Indeed, starting from one vertex, each successive closed neighborhood increases the size by at most a factor $\mu+1$. The union of the balls around $u,v$ therefore has at most $2(\mu+1)^g<\nu$ vertices. The chosen matching has $\nu$ disjoint edges, so some edge $xy$ of the same matching has both endpoints outside this union. All cross-distances between $\{u,v\}$ and $\{x,y\}$ exceed $g$.

Switch $uv,xy$ to $uy,xv$, preserving the color matching. No cycle of length at most $r$ can use a new edge. A cycle using exactly one new edge would leave an old cross-path of length at most $r-1\leq g$. If it uses both, deleting those edges leaves two old paths. One possible pairing contains a cross-path, again impossible. The other gives alternate paths around the old edges $uv$ and $xy$. Each has length at least $r-1$, because adding its missing old edge creates an old cycle. The new cycle would then have length at least $2(r-1)+2=2r>r$.

Thus every cycle of length at most $r$ after switching uses only retained old edges, while the selected $r$-cycle is gone. Either the truncated girth increases or the number of cycles of length $r$ decreases, contradicting the extremal choice. There is no cycle of length at most $g$.

Finally take a nonempty connected component. Every edge incident to one of its vertices stays in the component, so all local matching conditions survive. No new cycle appears, and neither side grows. This proves~\eqref{eq:coverbound} and the required girth bound.
\end{proof}

\section{The Recursive Auxiliary Input}\label{app:input}
This appendix proves Lemma~\ref{lem:input}. It uses the base and recursive substitution of HWZ~\cite[Section 3]{HWZ2020}, retaining the integer decoding identities and exact counts needed by the reversible compiler. The distance proof is in Appendix~\ref{app:excursions}.

\subsection{Base graph and exact local formulas}
For $u\geq5$, put $q=\lceil u\rceil$. Use distinct vertices
\[
 v,w,s_1,\ldots,s_q,t_1,\ldots,t_q
\]
and, for every ordered pair $i\ne j$, two further private vertices $x_{ij},y_{ij}$. Add the edges $s_iv$, $vw$, $wt_i$ and the private paths
\[
 s_i,x_{ij},y_{ij},t_j\qquad(i\ne j).
\]
With $F=\{vw\}$, the parameters are
\[
 k=q,\quad f=1,\quad m=2q+1+3q(q-1)=3q^2-q+1.
\]
The graph is finite, connected, and loopless, and its terminal vertices are all distinct.

At round zero, $s_i$ sends $X_i$ to $v$ and all $x_{ij}$. At round one, $v$ sends $\sum_iX_i$ to $w$, and $x_{ij}$ forwards $X_i$ to $y_{ij}$. At round two, $w$ sends the sum to every sink, and $y_{ij}$ forwards $X_i$ to $t_j$. Each operation uses only initial local data or strictly earlier receptions. Formula~\eqref{eq:basecode} is an integer decoder at every sink, so the protocol is exact over arbitrary additive abelian groups.

After deleting $vw$, the source star and sink star communicate only through the private side paths. A simple source-to-sink path traverses an odd number of side paths. Three or more cost at least nine edges. If there is one, it cannot run from $s_i$ directly to $t_i$, since that side path was omitted. Reaching the side path at a different source or leaving it at a different sink requires two star edges, so the total is at least five. Taking any $j\ne i$, the path
\[
 s_i,v,s_j,x_{ji},y_{ji},t_i
\]
has length five. Thus the deleted own-pair distance is exactly five. Cycle erasure extends the same lower bound to all walks.

\subsection{The recursive physical product}
Take outer and inner inputs with respective parameters
\[
 (a_1,b_1,k_1,f_1,m_1),\qquad
 (a_2,b_2,k_2,f_2,m_2),
\]
where $a_i$ is a full-duplex code duration and $b_i$ is a deleted own-pair distance bound. Suppose $m_1>f_1$. Index both orientations of each noncut outer edge by $d=2(m_1-f_1)$ colors. By Lemma~\ref{lem:cover}, choose a connected complementary colored cover with parameters $(d,k_2)$ and girth at least $2b_1b_2$. Write its side sizes as $n_L,n_R$; then
\begin{equation}\label{eq:inputincidence}
 n_Ld=n_Rk_2.
\end{equation}

Take an outer copy per left vertex and an inner copy per right vertex. For an incidence with colors $(e,i)$, identify the tail and head of outer directed arc $e$ with the source and sink of inner session $i$. Delete every noncut outer edge. Replace each cut outer edge by a private path of exactly $a_2$ edges. The demands are the copied outer pairs. The new cut set contains the first edge of each cut replacement and every cut edge in every inner copy. Consequently
\begin{equation}\label{eq:inputcounts}
 \begin{split}
 k_+&=n_Lk_1,\\
 f_+&=n_Lf_1+n_Rf_2,\\
 m_+&=a_2n_Lf_1+n_Rm_2.
 \end{split}
\end{equation}

There is no hidden identification between distinct outer vertices. To check this, first replace each identification by a zero-length connector. Each inner terminal has one connector, so every connector component is a star around one outer vertex, or an isolated inner nonterminal. Different terminals in one inner copy attach through different incidence edges, which have different left endpoints because the cover is simple. The two terminals attached through one incidence go to the two distinct endpoints of a loopless outer edge. It follows that the inner-copy maps are injective as well. Retained edges stay loopless. Parallel edge identities, if any, are kept distinct.

Connectivity follows from the same description. The ends of each removed outer edge remain connected through an assigned connected inner copy, and each cut replacement is a path. Thus the vertices of each outer copy remain mutually connected. Each inner copy meets all its incident outer copies, and the cover is connected. Its internal vertices and all replacement-path vertices belong to the resulting connected graph.

\subsection{Causal integer composition}
For each outer round, compute the outgoing payload on every directed noncut edge from the outer history available before that round. Assign the zero payload to a direction with no scheduled transmission in that round. Assign this payload to the corresponding inner source, which has been identified with that outer sender. Execute all inner copies for $a_2$ rounds. During the same interval, relay each cut-edge payload along its private $a_2$-edge path. The full-duplex convention allows the two opposite cut payloads to advance simultaneously on a segment.

By the end of the interval, every outer receiver has its intended payload. An inner decoder is an identity for all values of its inputs, so substituting correlated integer forms in the original messages is valid. Locally, substitute each assigned payload expression in the corresponding inner source term. That payload is known at the identified sender when the interval begins. Inner predecessor terms refer to earlier rounds of the current inner execution; outer-history terms refer to completed earlier intervals. Thus every transmission has a local strictly causal integer expression.

The edges of different inner copies and cut replacement paths are disjoint as physical identities, so the full-duplex budget is respected. Induction on outer rounds proves exactness. The code duration is $a_+=a_1a_2$.

After deleting the new cut set, the remaining parts of the cut replacements are stubs attached to at most one outer vertex. A walk between old terminals can remove its visits to these stubs. The remaining core is exactly the attachment construction of Appendix~\ref{app:excursions}, using the inner graphs with their cut sets deleted and the noncut outer edges. That argument proves
\begin{equation}\label{eq:inputdistanceproduct}
 b_+=b_1b_2.
\end{equation}
Here $b_+$ is a lower bound, not necessarily the exact distance.

\subsection{Maintaining the session-to-cut ratio}
Write $R=k/f$ and $U=m/f$. The base input has $R\geq u$ and $U\geq3$. Assume inputs exist at level $h$ for every rational target at least five. To produce level $h+1$ with target $u$, choose the outer input at target $3u$ and then the inner input at target $\max\{5,U_1\}$. Their code times are both $3^{2^h}$ and their distance lower bounds both $5^{2^h}$.

Using~\eqref{eq:inputincidence} and~\eqref{eq:inputcounts},
\begin{equation}\label{eq:inputratio}
 R_+=\frac{R_1}{1+2(U_1-1)/R_2}
 \geq\frac{3u}{1+2U_1/R_2}\geq u.
\end{equation}
Moreover,
\[
 U_+=\frac{a_2n_Lf_1+U_2n_Rf_2}{n_Lf_1+n_Rf_2}
\]
is a weighted average of $a_2$ and $U_2$, both at least three. Thus $U_+\geq3$, ensuring that the next product has positive $d$. Its cut set is nonempty. The duration and distance bound square at each step, yielding $3^{2^{h+1}}$ and $5^{2^{h+1}}$.

\subsection{One final serialization}
The recursion above stays in the full-duplex convention. Serialize only after it is complete. Choose one reference direction for each physical edge. Replace round $t$ by subrounds $2t,2t+1$, assigned respectively to the reference and reverse directions. Both payloads are computed from the history available before $2t$. An original predecessor at $t'<t$ is transmitted by $2t'+1<2t$, so all local expressions remain strictly causal. The resulting code has duration $2\cdot3^{2^h}$ and uses at most one symbol in total on each edge per subround. This proves Lemma~\ref{lem:input}.

\section{Undirected Excursions in Graph Products}\label{app:excursions}
This appendix gives the distance argument shared by the tensor product and the auxiliary input recursion. It is useful to keep the incidence graph separate from the physical graph: incidence edges describe attachments, whereas only inner physical edges count toward a route's hop length.

\subsection{Lifting a quotient path}
Start with disjoint outer vertex sets and inner graphs. Every incidence attaches one inner session to the endpoints of one loopless outer edge or arc, by two zero-length connectors. Each inner terminal has exactly one connector; inner nonterminals have none. After deleting the outer edges, quotienting the connectors gives the physical construction used in the paper.

Let an arbitrary physical walk join one own-pair in a fixed outer copy. Lift each visit to an identified outer vertex through its connector star, obtaining a walk in the expanded graph. Connectors have length zero and physical inner edges length one. Erase repeated-vertex portions. The result is a simple expanded path with no more physical edges than the original walk, and with the same two outer endpoints. Call its physical length $p$.

Map each outer vertex to its outer-copy label, and each inner vertex to its inner-copy label. A physical step remains at one inner label; a connector step traverses an incidence edge. Removing the stationary steps gives the incidence trace, a closed walk in the cover based at the outer copy containing the endpoints, which we call the root. Let $\kappa$ be its number of connector crossings.

Each visit from an outer copy into an inner copy and back uses a physical edge between the two connector crossings. Otherwise the simple path would leave the same inner terminal through its unique connector, immediately revisiting a vertex. The alternating pairs of crossings have disjoint intervening path portions, so
\begin{equation}\label{eq:crossings}
 \kappa\leq2p.
\end{equation}

\subsection{A cyclic incidence trace}
If the support of the incidence trace contains a graph-theoretic cycle, that cycle has length at least the incidence girth $2b_1b_2$. Its distinct edges all appear among the trace steps. Therefore $\kappa\geq2b_1b_2$, and~\eqref{eq:crossings} yields $p\geq b_1b_2$.

This case permits arbitrary backtracking and repeated copy visits in the incidence trace. The only simplicity used is that of the expanded physical path, which was obtained by cycle erasure before tracing.

\subsection{A tree-supported incidence trace}
Suppose the trace support is a tree and root it at the outer copy containing the endpoints. A closed walk in a tree crosses each used edge an even positive number of times. Every incidence edge has exactly two corresponding connectors in the expanded graph. A simple expanded path uses either connector at most once. It follows that each used incidence edge is crossed exactly twice, using its two different connectors.

In a rooted tree, these paired crossings are nested. For an edge from a parent copy to a child copy, the portion between its first and second crossing stays entirely in the child's subtree. Otherwise it would have to cross that edge again or create an incidence cycle. Call that portion the child excursion, including its two boundary connectors.

An excursion into an inner copy begins and ends at the two terminals of the session assigned to its parent incidence. An excursion into a nonroot outer copy begins and ends at the two different endpoints of the outer edge assigned to its parent incidence. The latter endpoints need not form a demanded pair; the proof uses only that they are distinct.

We show that every inner excursion costs at least $b_2$ physical edges, by induction on its finite expanded-path length. If it contains no child outer excursion, its physical part is a walk wholly inside its own inner graph between an own-pair, so the inner distance bound applies. Otherwise it contains an outer child excursion between distinct outer vertices. There are no physical edges left inside an outer copy. To connect its distinct endpoints, this child must contain an inner grandchild excursion. The grandchild is a strictly shorter nested path portion. By induction it costs at least $b_2$, and its physical edges are among those of the original inner excursion. This proves the bound without applying an own-pair assumption to the nonroot outer endpoints.

At the root outer copy, the expanded path consists of maximal inner excursions joined at root vertices. Each such excursion returns along the same incidence edge through its other connector. Its endpoints are precisely the two endpoints of that incidence's outer edge or arc. Contracting it to the corresponding undirected outer edge produces a walk from the root source to its own sink. The outer distance bound requires at least $b_1$ such edges. Maximal root excursions occupy disjoint portions of the simple path, and each costs at least $b_2$, so $p\geq b_1b_2$.

\subsection{Application with deleted edges}
For the scalar tensor of Section~\ref{sec:tensor}, use the full inner graphs and all outer edges. For the auxiliary recursion, first remove the new cut set. A cut replacement has lost its first edge, leaving only isolated vertices or a path attached to one outer endpoint. Deleting excursions into these stubs does not change the outer endpoints and can only shorten a walk. The inner pieces that remain are copies of the inner graph with its cut set deleted, and every root contraction uses a noncut outer edge. Thus the preceding proof applies with the deleted own-pair distances $b_1,b_2$.

The proof begins with an arbitrary undirected walk in either case. Both forward and reverse uses of any physical edge are covered, including routes that pass through inner copies in a way unrelated to the coding schedule.

\section{Organization of the Lean Proof}\label{app:formalization}
The formalization follows the mathematical construction. It represents a finite network together with its protocol and proves the properties needed by the next construction. The proof preserves the meaning of both communication models through the seed construction, tensor iteration, and subdivision.

\subsection{Representing the network and the two communication models}
A physical network is represented by finite vertex and edge sets $V,E$, with two distinct endpoints assigned to every edge. The endpoint labels provide a reference orientation; walks may traverse an edge in either direction. Keeping $E$ as a separate set is important in the auxiliary multigraphs, where two edges can join the same vertices but have separate capacity budgets. Simplicity is an additional property asserting that distinct physical edges have distinct unordered endpoint pairs. Connectivity, degree, terminal leaves, and cycle lengths are proved for the graph defined by these endpoints.

A finite coding protocol has a finite set of transmission events, each with a round, a physical edge, and a sending endpoint, together with the encoding data described next. Its capacity obligation states that two different events cannot occupy the same physical edge and round. The direction does not enter this condition, so the two directions of an edge share its budget. The scalar instances satisfy the stronger property that every physical edge occurs exactly once among all events. That stronger property is what the pipeline proof uses when several scalar executions overlap.

For an integer protocol, each payload is represented by its coefficient row in the original input coordinates. A row gives the symbol's global value, but it does not by itself show that the sender can compute it. For each event the formalization therefore records a local expression and establishes three facts about it: every source coordinate used is located at the sender; every received symbol used arrived at that sender in a strictly earlier round; and the resulting integer row equals the payload row. A decoder has corresponding local-source and received-symbol obligations, together with the identity that its output row is the desired input-coordinate row.

These identities are proved over $\mathbb Z$. Evaluating a row on group-valued inputs means taking the integer-weighted sum of the input coordinates. Evaluation preserves addition and integer multiplication, so the local row identities become valid local symbol operations in every additive abelian group. Over a field, the same evaluation is a linear map over that field. The proof checks the linearity of the local encoders and decoders as functions of their available data, as required by Section~\ref{sec:model}.

The upper-bound coding model separately permits arbitrary deterministic local functions. Their arguments contain only source coordinates present at the sender and symbols received there before the event. Thus equality of these local views implies equality of the next transmitted symbol. This is the induction principle used in the source-leaf transcript argument. The finite alphabet is needed when injective decoding turns a bound on transcript length into the inequality $B_i\leq D$.

Fractional routing is represented independently of the coding schedule. For each session, a real-valued function assigns a weight to every simple undirected source--sink path. The definition requires nonnegative weights, sufficient total weight for each demand, and the shared physical-edge load bound~\eqref{eq:routing}. Paths are sequences of physical edges, so parallel edges remain distinct. The proof of cycle erasure preserves each demand and does not increase any edge's traversal count; the set of all simple paths is proved finite, so sums over it are well defined. The metric theorem is then the finite-sum inequality~\eqref{eq:lengthbound}, applied to these path weights.

\subsection{Constructing the seed and its certificate}
The initial construction produces the auxiliary graph and its integer code together. In the base graph, the source coefficients and decoder rows implement~\eqref{eq:basecode}. For the recursive product, the inner code is instantiated on the outer payload rows. A substitution lemma proves that decoding identities remain true when the inner inputs are correlated. The same lemma preserves the sender locations, receiver locations, and strict time offsets, so a correct algebraic substitution also gives a legal local protocol.

The reversible compiler is verified in two stages. First, induction over the causal schedule proves the clean-input identity $LUJ=\Id_k$. Second, the cancellation calculation in Section~\ref{sec:compiler} is proved as an identity on the full register space. Expanding the maps into integer elementary gates gives the signed exchange~\eqref{eq:signedswap}, including preservation of every arbitrary workspace coordinate. The compiler also records the host of every register and the physical edge of every cross gate. Reversal, transposition, and coefficient negation preserve the cross-phase restrictions used in the routing argument.

The graph of register strands is constructed from the gate word. Splitting its gates yields the subcubic graph with explicit endpoints, and a chronological row evaluation proves that the code transmitted on those edges is the circuit computation. The temporal metric proof covers all undirected walks through this graph. It establishes the baseline equality, the marked-port budget, and the extra distance paid by each designated demand. These facts yield the strict integer surplus~\eqref{eq:strictmetric}.

The subsequent transformations retain both a code and a distance certificate. Positive metric repair uses the integrality of the surplus. Component selection uses a sum over components to find one with positive surplus; exact decoding proves that retained source and sink endpoints lie together. Source-leaf padding equalizes the distance lower bounds. Variable subdivision supplies private replacement paths, a relay schedule, and the resulting unweighted hop bound. A relabeling of vertices, edges, and terminals then presents this graph as a scalar instance in the sense of Definition~\ref{def:scalar}, with its physical properties unchanged. The seed therefore comes with the strict inequality $m_0<k_0b_0$. The block pipeline turns this integer surplus into the finite rate comparison~\eqref{eq:finitestrict}; evaluating the local rows in $\mathbb F_2$ gives the finite counterexample.

Some finite choices, such as the component or the cover, are made by classical existence arguments; their cardinalities and durations are derived from the construction.

\subsection{Composing the amplified family}
The tensor operation takes two scalar instances and constructs a third. Its vertex set consists of copied outer vertices and private inner nonterminals; its edge set consists of copied inner physical edges. With this representation, simplicity, connectivity, degree preservation, and leaf-terminal preservation are proved separately. The code proof assigns an inner execution to each outer edge event and verifies local availability at the chosen time offsets. The distance proof translates undirected walks into the incidence excursions of Appendix~\ref{app:excursions}.

The colored cover is also proved to exist within the development. The argument optimizes over a finite space of color-preserving matchings by an extremal choice. A switching lemma removes a shortest cycle without introducing an equally short new cycle, and a neighborhood bound guarantees a distant edge for the switch. These facts establish both the required girth and the numerical side-size bound. The tensor proof can therefore use the cover together with its proved size bound at every recursive step.

Iteration supplies scalar instances, not only a recurrence for their proposed parameters. The exact edge and session counts prove that the certified ceiling squares. The cover bound gives the edge-growth inequality, from which the logarithmic exponent is derived. Uniform subdivision is likewise an operation on a graph together with its code. Its endpoint construction proves bipartiteness for even replacement length and the lower bound on every physical cycle. The final real inequalities choose the replacement length and transfer the logarithmic bound to the new vertex count.

The coding-capacity and vector-region results are attached to every scalar instance. For each positive block length, the scalar protocol yields a block protocol with one-round start spacing. In the formalized family theorem this is recorded in threshold form: for every rate below one there is a block length whose protocol exceeds it; the protocol for each individual block length is proved separately. The finite-block separation is also proved directly. For every network in the family there is a block length whose protocol has common rate strictly above every feasible fractional routing rate, which is item~5 of Theorem~\ref{thm:main}. For the seed network, the integer protocol at such a block length, interpreted over the two-element field, is a linear causal code, which is Theorem~\ref{thm:finite}. The exact finite-block region of Proposition~\ref{prop:exactregion} and the near-linear-girth family of Corollary~\ref{cor:nearlinear} are proved in an extension of the development that is audited in the same way. The source-leaf proof bounds all deterministic symbol rates by one. For any target vector in the unit cube and any positive error tolerance, coordinate restriction and~\eqref{eq:vectorerror} produce a finite protocol within that tolerance. These lower and upper bounds prove the least-upper-bound and closure statements used by the family theorem.

The final assembly first obtains the integer seed and its strict ceiling, then fixes the positive logarithmic exponent. For each $\beta\in(0,1)$, it constructs the iterated and subdivided family and its prefactor. Only after these graph and protocol choices does it interpret the integer instructions over a group or field. This order is the source of the simultaneous alphabet quantifiers in Theorem~\ref{thm:main}.

\subsection{What the verification establishes}
Lean checks that each constructed object satisfies its declared obligations and that the final theorem follows from them. Its kernel checks proof terms generated by the development, including the terms produced with arithmetic and algebraic tactics. The checked theorem is relative to the formal definitions. Their interpretation as the communication model in Section~\ref{sec:model} is the mathematical correspondence established by the representations and operational arguments above.

The development builds without errors. An audit of theorem dependencies over its 2,472 declarations and 1,755 theorems shows that the proofs use only Lean's standard axioms (propositional extensionality, classical choice, and quotient soundness), with no further axioms and no incomplete proofs. An independent re-check of all logical declarations in a fresh kernel environment also succeeded; it uses the same Lean kernel rather than an independent verifier.

\begingroup
\interlinepenalty=10000
\bibliography{references}

@IEEEtranBSTCTL{IEEEexample:BSTcontrol,
  CTLdash_repeated_names = "no"
}

@inproceedings{LiLi2004,
  author = {Zongpeng Li and Baochun Li},
  title = {Network Coding: The Case of Multiple Unicast Sessions},
  booktitle = {Proceedings of the 42nd Annual Allerton Conference on Communication, Control, and Computing},
  year = {2004}
}

@inproceedings{HWZ2020,
  author = {Bernhard Haeupler and David Wajc and Goran Zuzic},
  title = {Network Coding Gaps for Completion Times of Multiple Unicasts},
  booktitle = {Proceedings of the 61st IEEE Annual Symposium on Foundations of Computer Science (FOCS)},
  year = {2020},
  pages = {494--505},
  publisher = {IEEE},
  doi = {10.1109/FOCS46700.2020.00053}
}

@misc{Liu2026,
  author = {Sirui Liu and Li Que and Zongpeng Li and Baochun Li},
  title = {On the Multiple-Unicast Conjecture: Beyond Cut Metrics},
  year = {2026},
  howpublished = {arXiv:2608.06070}
}

@techreport{HKL2004,
 author = {Nicholas J. A. Harvey and Robert D. Kleinberg and April {Rasala Lehman}},
 title = {Comparing Network Coding with Multicommodity Flow for the $k$-Pairs Communication Problem},
 institution = {Massachusetts Institute of Technology},
 year = {2004},
 month = nov,
 number = {MIT-LCS-TR-964}
}

@inproceedings{BGS2017,
 author = {Mark Braverman and Sumegha Garg and Ariel Schvartzman},
 title = {Coding in Undirected Graphs Is Either Very Helpful or Not Helpful at All},
 booktitle = {8th Innovations in Theoretical Computer Science Conference (ITCS 2017)},
 series = {Leibniz International Proceedings in Informatics (LIPIcs)},
 volume = {67},
 pages = {18:1--18:18},
 year = {2017},
 publisher = {Schloss Dagstuhl -- Leibniz-Zentrum f{\"u}r Informatik},
 doi = {10.4230/LIPIcs.ITCS.2017.18}
}

@article{HKL2006,
  author = {Nicholas J. A. Harvey and Robert D. Kleinberg and April {Rasala Lehman}},
  title = {On the Capacity of Information Networks},
  journal = {IEEE Transactions on Information Theory},
  volume = {52},
  number = {6},
  pages = {2345--2364},
  year = {2006},
  month = jun,
  doi = {10.1109/TIT.2006.874531}
}

@article{JVY2006,
  author = {Kamal Jain and Vijay V. Vazirani and Gideon Yuval},
  title = {On the Capacity of Multiple Unicast Sessions in Undirected Graphs},
  journal = {IEEE Transactions on Information Theory},
  volume = {52},
  number = {6},
  pages = {2805--2809},
  year = {2006},
  month = jun,
  doi = {10.1109/TIT.2006.874543}
}

@inproceedings{AY2008,
  author = {Ali Al-Bashabsheh and Abbas Yongacoglu},
  title = {On the $k$-Pairs Problem},
  booktitle = {Proceedings of the IEEE International Symposium on Information Theory (ISIT)},
  year = {2008},
  pages = {1828--1832},
  doi = {10.1109/ISIT.2008.4595304}
}

@inproceedings{LM2009,
  author = {Michael Langberg and Muriel M{\'e}dard},
  title = {On the Multiple Unicast Network Coding Conjecture},
  booktitle = {Proceedings of the 47th Annual Allerton Conference on Communication, Control, and Computing},
  year = {2009},
  pages = {222--227},
  doi = {10.1109/ALLERTON.2009.5394800}
}

@article{Xiahou2014,
  author = {Tang Xiahou and Zongpeng Li and Chuan Wu and Jiaqing Huang},
  title = {A Geometric Perspective to Multiple-Unicast Network Coding},
  journal = {IEEE Transactions on Information Theory},
  volume = {60},
  number = {5},
  pages = {2884--2895},
  year = {2014},
  month = may,
  doi = {10.1109/TIT.2014.2308998}
}

@inproceedings{TQ2019,
  author = {Satyajit Thakor and Mohammad Ishtiyaq Qureshi},
  title = {Undirected Unicast Network Capacity: A Partition Bound},
  booktitle = {Proceedings of the IEEE International Symposium on Information Theory (ISIT)},
  year = {2019},
  pages = {1--5},
  doi = {10.1109/ISIT.2019.8849499}
}

@article{Hu1963,
  author = {T. C. Hu},
  title = {Multi-Commodity Network Flows},
  journal = {Operations Research},
  volume = {11},
  number = {3},
  pages = {344--360},
  year = {1963},
  doi = {10.1287/opre.11.3.344}
}

@inproceedings{Liu2026sif,
  author = {Sirui Liu and Zongpeng Li and Xiying Fan and Haifeng Chen},
  title = {A Session Interaction Framework for the Multiple-Unicast Conjecture},
  booktitle = {Proceedings of the IEEE International Symposium on Information Theory (ISIT)},
  year = {2026},
  pages = {1--6},
  doi = {10.1109/ISIT62367.2026.11653955}
}

@article{LLR1995,
  author = {Nathan Linial and Eran London and Yuri Rabinovich},
  title = {The Geometry of Graphs and Some of Its Algorithmic Applications},
  journal = {Combinatorica},
  volume = {15},
  number = {2},
  pages = {215--245},
  year = {1995},
  doi = {10.1007/BF01200757}
}

@article{AR1998,
  author = {Yonatan Aumann and Yuval Rabani},
  title = {An {$O(\log k)$} Approximate Min-Cut Max-Flow Theorem and Approximation Algorithm},
  journal = {SIAM Journal on Computing},
  volume = {27},
  number = {1},
  pages = {291--301},
  year = {1998},
  doi = {10.1137/S0097539794285983}
}

@inproceedings{Lean4,
  author = {Leonardo de Moura and Sebastian Ullrich},
  title = {The {Lean 4} Theorem Prover and Programming Language},
  booktitle = {Automated Deduction---CADE 28},
  series = {Lecture Notes in Computer Science},
  volume = {12699},
  pages = {625--635},
  year = {2021},
  publisher = {Springer},
  doi = {10.1007/978-3-030-79876-5_37}
}

@inproceedings{Mathlib2020,
  author = {{The mathlib Community}},
  title = {The {Lean} Mathematical Library},
  booktitle = {Proceedings of the 9th ACM SIGPLAN International Conference on Certified Programs and Proofs (CPP 2020)},
  pages = {367--381},
  year = {2020},
  publisher = {ACM},
  doi = {10.1145/3372885.3373824}
}

@article{ACLY2000,
  author = {Rudolf Ahlswede and Ning Cai and Shuo-Yen Robert Li and Raymond W. Yeung},
  title = {Network Information Flow},
  journal = {IEEE Transactions on Information Theory},
  volume = {46},
  number = {4},
  pages = {1204--1216},
  year = {2000},
  month = jul,
  doi = {10.1109/18.850663}
}

@inproceedings{LiLiCISS2004,
  author = {Zongpeng Li and Baochun Li},
  title = {Network Coding in Undirected Networks},
  booktitle = {Proceedings of the 38th Annual Conference on Information Sciences and Systems (CISS)},
  year = {2004},
  month = mar
}

@article{OS1981,
  author = {Haruko Okamura and Paul D. Seymour},
  title = {Multicommodity Flows in Planar Graphs},
  journal = {Journal of Combinatorial Theory, Series B},
  volume = {31},
  number = {1},
  pages = {75--81},
  year = {1981},
  doi = {10.1016/S0095-8956(81)80012-3}
}

@article{KS2006,
  author = {Gerhard Kramer and Serap A. Savari},
  title = {Edge-Cut Bounds on Network Coding Rates},
  journal = {Journal of Network and Systems Management},
  volume = {14},
  number = {1},
  pages = {49--67},
  year = {2006},
  month = mar,
  doi = {10.1007/s10922-005-9019-0}
}

@article{Yin2018,
  author = {Xunrui Yin and Zongpeng Li and Yaduo Liu and Xin Wang},
  title = {A Reduction Approach to the Multiple-Unicast Conjecture in Network Coding},
  journal = {IEEE Transactions on Information Theory},
  volume = {64},
  number = {6},
  pages = {4530--4539},
  year = {2018},
  month = jun,
  doi = {10.1109/TIT.2017.2771738}
}

@article{QT2022,
  author = {Mohammad Ishtiyaq Qureshi and Satyajit Thakor},
  title = {A Bound on Undirected Multiple-Unicast Network Information Flow},
  journal = {IEEE Transactions on Information Theory},
  volume = {68},
  number = {7},
  pages = {4453--4469},
  year = {2022},
  month = jul,
  doi = {10.1109/TIT.2022.3157623}
}

@inproceedings{LiuINFOCOM2026,
  author = {Sirui Liu and Yeqiao Hou and Jessie Hui Wang and Zongpeng Li},
  title = {On the Multiple-Unicast Conjecture: Session Dominance},
  booktitle = {Proceedings of IEEE INFOCOM},
  year = {2026},
  doi = {10.1109/INFOCOM59046.2026.11571518}
}

@article{Bennett1973,
  author = {Charles H. Bennett},
  title = {Logical Reversibility of Computation},
  journal = {IBM Journal of Research and Development},
  volume = {17},
  number = {6},
  pages = {525--532},
  year = {1973},
  month = nov,
  doi = {10.1147/rd.176.0525}
}

@inproceedings{FHLS2019,
  author = {Alireza Farhadi and MohammadTaghi Hajiaghayi and Kasper Green Larsen and Elaine Shi},
  title = {Lower Bounds for External Memory Integer Sorting via Network Coding},
  booktitle = {Proceedings of the 51st Annual ACM SIGACT Symposium on Theory of Computing (STOC)},
  year = {2019},
  pages = {997--1008},
  doi = {10.1145/3313276.3316337}
}

@inproceedings{AFKL2019,
  author = {Peyman Afshani and Casper Benjamin Freksen and Lior Kamma and Kasper Green Larsen},
  title = {Lower Bounds for Multiplication via Network Coding},
  booktitle = {Proceedings of the 46th International Colloquium on Automata, Languages, and Programming (ICALP)},
  series = {Leibniz International Proceedings in Informatics (LIPIcs)},
  volume = {132},
  publisher = {Schloss Dagstuhl -- Leibniz-Zentrum f{\"u}r Informatik},
  pages = {10:1--10:12},
  year = {2019},
  doi = {10.4230/LIPIcs.ICALP.2019.10}
}

@inproceedings{BH2025,
  author = {Mark Braverman and Zhongtian He},
  title = {Undirected Multicast Network Coding Gaps via Locally Decodable Codes},
  booktitle = {Proceedings of the 66th IEEE Annual Symposium on Foundations of Computer Science (FOCS)},
  year = {2025}
}
\endgroup
\end{document}